\documentclass[10pt]{amsart}
\usepackage{epsf,latexsym,amsmath,amssymb,amscd,datetime}
\usepackage{amsmath,amsthm,amssymb,enumerate,eucal,url,calligra,mathrsfs}

\usepackage{blkarray} 

\usepackage{tikz,scalerel}
\usetikzlibrary{arrows}

\usepackage{imakeidx}     

\usepackage[outdir=./]{epstopdf}

\usepackage{graphicx}
\usepackage{color}
\newenvironment{jfnote}{ \bgroup \color{blue} }{\egroup}

\newcommand{\red}{\color[rgb]{1.0,0.2,0.2}} 

\newcommand{\oldStuff}[1]{}

\newcommand{\Coord}{{\rm Coord}}
 
\newcommand{\og}{{\scriptscriptstyle \le}}
\newcommand{\Bg}{{\scalebox{1.0}{$\!\scriptscriptstyle /\!B$}}}

\DeclareMathOperator{\SHom}{\mathscr{H}\text{\kern -3pt {\calligra\large om}}\,}

\IfFileExists{my_xrefs}{\input my_xrefs}{}

\DeclareMathOperator{\ViSu}{VisSub}

\newcommand{\naturals}{{\mathbb N}}

\newcommand{\Eor}{E^{\mathrm{or}}}
\newcommand{\mec}[1]{{\bf #1}}	
\newcommand{\bec}[1]{{\boldsymbol #1}}	

\newcommand{\specnew}{\Spec^{\mathrm{new}}}

\usepackage{mathrsfs}
\usepackage{amssymb}
\usepackage{dsfont}
\usepackage{verbatim}
\usepackage{url}

\newcommand{\Edir}{E^{\mathrm{dir}}}

\theoremstyle{plain}
\newtheorem{theorem}{Theorem}[section]
\newtheorem{lemma}[theorem]{Lemma}

\newtheorem{corollary}[theorem]{Corollary}

\theoremstyle{definition}
\newtheorem{definition}[theorem]{Definition}

\newtheorem{xca}{Exercise}[section]

\newcommand{\isom}{\simeq} 

\newcommand{\ignore}[1]{}

\newcommand{\reals}{{\mathbb R}}

\newcommand{\integers}{{\mathbb Z}}

\newcommand{\complex}{{\mathbb C}}

\newcommand\EE{\mathbb{E}}

\newcommand\II{\mathbb{I}}

\DeclareMathAlphabet{\mathcal}{OMS}{cmsy}{m}{n}

\newcommand\cC{\mathcal{C}}

\newcommand\cP{\mathcal{P}}

\newcommand\cR{\mathcal{R}}
\newcommand\cS{\mathcal{S}}
\newcommand\cT{\mathcal{T}}

\DeclareMathOperator{\Prob}{Prob}

\DeclareMathOperator{\VLG}{VLG}

\DeclareMathOperator{\Line}{Line}

\DeclareMathOperator{\SNBC}{SNBC}

\DeclareMathOperator{\snbc}{snbc}

\def\from{\colon}

\def\isom{\simeq}

\def\eqdef{\overset{\text{def}}{=}}

\DeclareMathOperator{\ord}{ord}

\DeclareMathOperator{\Spec}{Spec}

\DeclareRobustCommand
  \rddots{\mathinner{\mkern1mu\raise\p@
    \vbox{\kern7\p@\hbox{.}}\mkern2mu
    \raise4\p@\hbox{.}\mkern2mu\raise7\p@\hbox{.}\mkern1mu}}

\newcommand\xhookrightarrow[2][]{\ext@arrow 0062{\hookrightarrowfill@}{#1}{#2}}
\def\hookrightarrowfill@{\arrowfill@\lhook\relbar\rightarrow}

\usepackage{relsize}
\usepackage{tikz}
\usetikzlibrary{matrix,arrows,decorations.pathmorphing}
\usepackage{tikz-cd}
\usetikzlibrary{cd}

\usepackage[pdftex,colorlinks,linkcolor=blue]{hyperref}

\newcommand{\myDeleteNote}[1]{{\red #1}}

\begin{document}

\title[Relativized Alon Conjecture VI] 
{On the Relativized Alon Second Eigenvalue
Conjecture VI: Sharp Bounds for Ramanujan Base Graphs}

\author{Joel Friedman}
\address{Department of Computer Science,
        University of British Columbia, Vancouver, BC\ \ V6T 1Z4, CANADA}
\curraddr{}
\email{{\tt jf@cs.ubc.ca}}
\thanks{Research supported in part by an NSERC grant.}

\author{David Kohler}
\address{Department of Mathematics,
        University of British Columbia, Vancouver, BC\ \ V6T 1Z2, CANADA}
\curraddr{422 Richards St, Suite 170, Vancouver BC\ \  V6B 2Z4, CANADA}
\email{{David.kohler@a3.epfl.ch}}
\thanks{Research supported in part by an NSERC grant.}

%
\date{\today}

\subjclass[2010]{Primary 68R10}

\keywords{}

\begin{abstract}

This is the sixth in a series of articles devoted to showing that a typical
covering map of large degree to a fixed, regular graph has its new adjacency
eigenvalues within the bound conjectured by Alon for random regular graphs.

In this article we show that if the fixed graph is regular Ramanujan,
then the {\em algebraic power} of the model of random
covering graphs is $+\infty$.  This
implies a number of
interesting results, such as (1) one obtains the upper
and lower bounds---matching to within a multiplicative constant---for
the probability that a random covering
map has some new adjacency eigenvalue outside the Alon bound, and
(2) with probability smaller than any negative power of the
degree of the covering map, 
some new eigenvalue fails to be within the Alon bound
without the covering map containing one of finitely many ``tangles''
as a subgraph (and this tangle containment event has low probability).

\end{abstract}

\maketitle
\setcounter{tocdepth}{3}
\tableofcontents

\newcommand{\sePrelimProofs}{17}


\section{Introduction}

This paper is the sixth in a series of six articles whose main
results are to prove a relativized version of Alon's Second
Eigenvalue Conjecture,
conjectured in \cite{friedman_relative}, in the case where
the base graph is regular.

The relativized Alon conjecture for regular base graphs was proven
in Article~V (i.e., the fifth
article in this series).
In this article we give a sharper version of the
relativized Alon conjecture
that holds for all of our 
{\em basic models} of a random covering map of degree $n$ to a fixed
base graph, $B$, provided that $B$ is $d$-regular and Ramanujan.
Roughly speaking, for a fixed such $B$,
and for a random covering map $G\to B$ of degree $n$,
for $n$ large
we determine---to within a constant
factor independent of $n$---the probability that this map fails to be
a relative expander, in the sense that its new spectral radius 
is larger than the
bound conjectured by Alon for random $d$-regular graphs;
this probability is proportional to a negative power of $n$
which we call the {\em tangle power} of the model.

Curiously, in \cite{friedman_alon} such upper and lower bounds
were established
for random $d$-regular graphs formed from $d/2$ permutations
(for $d$ even) for all $d$ except those that 
are one more than a perfect odd square (e.g., $10,26,50,\ldots$).
However, the upper and lower bounds for these exceptional $d$ differed
by a factor of $n$ in \cite{friedman_alon}, 
and the results in this article close this bound
(since such random graphs are included in our
{\em basic models}, where the base graph, $B$, is a bouquet of whole-loops and
easily seen to be Ramanujan).

In Article~V we proved that the probability that a random
covering graph has a new eigenvalue outside the Alon bound
is bounded above proportional to
$n^{-\tau_1}$ and below proportional to $n^{-\tau_2}$, where
$$
\tau_1 = \min(\tau_{\rm tang},\tau_{\rm alg}), \quad
\tau_2 = \min(\tau_{\rm tang},\tau_{\rm alg}+1),
$$
where $\tau_{\rm tang}$ is a positive integer 
and $\tau_{\rm alg}$ is either a positive integer or $+\infty$
(both depending on the base graph, $B$, and the model
of random covering map).
However,
the integer $\tau_{\rm alg}$ appears to be very difficult to compute
directly:
there is---in principle---a finite algorithm to determine if
$\tau_{\rm alg}$ is larger than any given integer,
but (1) we know of no finite algorithm to check that $\tau_{\rm alg}=+\infty$,
and (2) when $\tau_{\rm alg}$ is larger than $1$ or $2$ the
direct computation of $\tau_{\rm alg}$ seems quite laborious.
On the other hand, the integer $\tau_{\rm tang}$ has a simple
meaning and is much easier to compute in practice.

In this article we show that for all of our {\em basic models}
of random covering maps of a $d$-regular Ramanujan
graph, $B$, $\tau_{\rm alg}=+\infty$;
the method of this proof goes back to \cite{friedman_random_graphs},
which uses the fact that $\tau_{\rm alg}$---at least when $B$ is
Ramanujan---is the order of the first
coefficient of an asymptotic expansion involving traces that 
grows as an exponential function with base $(d-1)^{1/2}$.
So rather than compute these asymptotic expansions directly, we use
the existence of these coefficients and apply other facts about
random graphs---namely Alon's notion of {\em magnification}---and
standard counting arguments to infer that the 
growth rates
of these asymptotic expansion coefficients
are strictly less than $(d-1)^{1/2}$.
As a consequence, we prove that 
$\tau_{\rm alg}=+\infty$ (without directly computing
asymptotic expansion coefficients); hence to determine
$\tau_1$ and $\tau_2$ above we need compute only $\tau_{\rm tang}$.

Once we formally define $\tau_{\rm alg}$, it becomes clear that 
$\tau_{\rm alg}=+\infty$ implies something quite strong for 
a $d$-regular $B$: namely, the
probability of having a new eigenvalue outside the Alon bound---namely,
larger
than $2(d-1)^{1/2}+\epsilon$ in absolute value
for any fixed $\epsilon>0$---can
be made smaller than any positive power of $n$, provided that
we discard the covering maps that contain certain {\em tangles}
(which are graph theoretically
local events that occur with probability 
proportional to $n^{-\tau_{\rm tang}}$).

Beyond our theorems in this article,
we conjecture that for our ``basic models'' of covering maps to
a fixed graph $B$ (regular or not), $\tau_{\rm alg}=+\infty$
($\tau_{\rm alg}$ and $\tau_{\rm tang}$ are defined for any
$B$, regular or not).

The rest of this article is organized as follows.
In Section~\ref{se_defs_review}, we review the definitions we will
need in this article; for more details, see Article~I in this series.
In Section~\ref{se_main_Ramanujan} we state the main theorems in this
article, and quote the results we will need from Article~V.
In Section~\ref{se_new_magnifiers} we review Alon's notion of 
{\em magnification} and introduce a variant of this notion,
{\em pseudo-magnification}, that will be useful to us.
In Section~\ref{se_pseudomag_nohalf} we prove that our basic
models of random covering maps to a base graph $B$ are
pseudo-magnifying in the case where $B$ has no half-loops;
this proof is computationally simpler that the general case,
although it illustrates all the main ideas.
In Section~\ref{se_pseudomag_with_half} we prove pseudo-magnification
for our basic models over general $B$.
In Section~\ref{se_proof_Ram} we use the pseudo-magnification
results to prove our main theorem, that $\tau_{\rm alg}=+\infty$
if $B$ is regular and Ramanujan.
In this case the probability of a cover having new adjacency eigenvalues
of absolute value outside the Alon bound is controlled by
$\tau_{\rm tang}$; we devote
Section~\ref{se_tau_tangle} to proving estimates on
$\tau_{\rm tang}$ for our basic models when $B$ is $d$-regular.

\section{Review of the Main Definitions}
\label{se_defs_review}

We refer the reader to Article~I for the definitions used in this article,
the motivation of such definitions, and an appendix there that lists all the
definitions and notation.
In this section we briefly review these definitions and notation. 

\subsection{Basic Notation and Conventions}
\label{su_very_basic}

We use $\reals,\complex,\integers,\naturals$
to denote, respectively, the
the real numbers, the complex numbers, the integers, and positive
integers or
natural numbers;
we use $\integers_{\ge 0}$ ($\reals_{>0}$, etc.)
to denote the set of non-negative
integers (of positive real numbers, etc.).
We denote $\{1,\ldots,n\}$ by $[n]$.

If $A$ is a set, we use $\naturals^A$ to denote the set of
maps $A \to \naturals$; we will refers to its elements as
{\em vectors}, denoted in bold face letters, e.g., $\mec k\in \naturals^A$
or $\mec k\from A\to\naturals$; we denote its {\em component}
in the regular face equivalents, i.e., for $a\in A$,
we use $k(a)\in\naturals$ to denote
the $a$-component of $\mec k$.
As usual, $\naturals^n$ denotes $\naturals^{[n]}=\naturals^{\{1,\ldots,n\}}$.
We use similar conventions for $\naturals$ replaced by $\reals$,
$\complex$, etc.

If $A$ is a set, then $\# A$ denotes the cardinality of $A$.
We often denote a set with all capital letters, and its cardinality
in lower case letters; for example,
when we define
$\SNBC(G,k)$, we will write
$\snbc(G,k)$ for $\#\SNBC(G,k)$.

If $A'\subset A$ are sets, then $\II_{A'}\from A\to\{0,1\}$ (with $A$
understood) denotes
the characteristic function of $A'$, i.e., $\II_{A'}(a)$ is $1$ if
$a\in A'$ and otherwise is $0$;
we also write $\II_{A'}$ (with $A$ understood) to mean $\II_{A'\cap A}$
when $A'$ is not necessarily a subset of $A$.

All probability spaces are finite; hence a probability space
is a pair $\cP=(\Omega,P)$ where $\Omega$ is a finite set and
$P\from \Omega\to\reals_{>0}$ with $\sum_{\omega\in\Omega}P(\omega)=1$;
hence an {\em event} is any subset of $\Omega$.
We emphasize that $\omega\in\Omega$ implies that $P(\omega)>0$ with
strict inequality; we refer to the elements of $\Omega$ as
the atoms of the probability space.
We use $\cP$ and $\Omega$ interchangeably when $P$ is
understood and confusion is unlikely.

A {\em complex-valued random variable} on $\cP$ or $\Omega$
is a function $f\from\Omega\to\complex$, and similarly for real-,
integer-, and natural-valued random variable; we denote its
$\cP$-expected value by
$$
\EE_{\omega\in\Omega}[f(\omega)]=\sum_{\omega\in\Omega}f(\omega)P(\omega).
$$
If $\Omega'\subset\Omega$ we denote the probability of $\Omega'$ by
$$
\Prob_{\cP}[\Omega']=\sum_{\omega\in\Omega'}P(\omega')
=
\EE_{\omega\in\Omega}[\II_{\Omega'}(\omega)].
$$
At times we write $\Prob_{\cP}[\Omega']$ where $\Omega'$ is
not a subset of $\Omega$, by which we mean
$\Prob_{\cP}[\Omega'\cap\Omega]$.

\subsection{Graphs, Our Basic Models, Walks}

A {\em directed graph},
or simply a {\em digraph},
is a tuple $G=(V_G,\Edir_G,h_G,t_G)$ consisting of sets
$V_G$ and $\Edir_G$ (of {\em vertices} and {\em directed edges}) and maps
$h_G,t_G$ ({\em heads}
and {\em tails}) $\Edir_G\to V_G$.
Therefore our digraphs can have multiple edges and
self-loops (i.e., $e\in\Edir_G$ with $h_G(e)=t_G(e)$).
A {\em graph} is a tuple $G=(V_G,\Edir_G,h_G,t_G,\iota_G)$
where $(V_G,\Edir_G,h_G,t_G)$ is a digraph and
$\iota_G\from \Edir_G\to \Edir_G$ is an involution with
$t_G\iota_G=h_G$;
the {\em edge set} of $G$, denoted $E_G$, is the
set of orbits of $\iota_G$, which (notation aside)
can be identified with $\Edir_G/\iota_G$,
the set of equivalence classes of
$\Edir_G$ modulo $\iota_G$;
if $\{e\}\in E_G$ is a singleton, then necessarily $e$ is a self-loop
with $\iota_G e =e $, and
we call $e$ a {\em half-loop}; other elements of $E_G$ are sets
$\{e,\iota_G e\}$ of size two, i.e., with $e\ne\iota_G e$, and for such $e$
we say that $e$ (or, at times, $\{e,\iota_G e\}$)
is a {\em whole-loop} if
$h_G e=t_G e$ (otherwise $e$ has distinct endpoints).

Hence these definitions allow our graphs to have multiple edges and 
two types of self-loops---whole-loops
and half-loops---as in
\cite{friedman_geometric_aspects,friedman_alon}.
The {\em indegree} and {\em outdegree} of a vertex in a digraph is
the number of edges whose tail, respectively whose head, is the vertex;
the {\em degree} of a vertex in a graph is its indegree (which equals
its outdegree) in the underlying digraph; 
therefore a whole-loop about a vertex contributes $2$
to its degree, whereas a half-loop contributes $1$.

An {\em orientation} of a graph, $G$, is a choice $\Eor_G\subset\Edir_G$
of $\iota_G$ representatives; i.e., $\Eor_G$ contains every half-loop, $e$,
and one element of each two-element set $\{e,\iota_G e\}$.

A {\em morphism $\pi\from G\to H$} of directed graphs is a pair
$\pi=(\pi_V,\pi_E)$ where $\pi_V\from V_G\to V_H$ and
$\pi_E\from \Edir_G\to\Edir_H$ are maps that intertwine the heads maps
and the tails maps of $G,H$ in the evident fashion;
such a morphism is {\em covering} (respectively, {\em \'etale},
elsewhere called an {\em immersion}) if for each $v\in V_G$,
$\pi_E$ maps those directed edges whose head is $v$ bijectively
(respectively, injectively) to those whose head is $\pi_V(v)$,
and the same with tail replacing head.
If $G,H$ are graphs, then a morphism $\pi\from G\to H$ is a morphism
of underlying directed graphs where $\pi_E\iota_G=\iota_H\pi_E$;
$\pi$ is called {\em covering} or {\em \'etale} if it is so as a morphism
of underlying directed graphs.
We use the words {\em morphism} and {\em map} interchangeably.

A walk in a graph or digraph, $G$, is an alternating sequence
$w=(v_0,e_1,\ldots,e_k,v_k)$ of vertices and directed edges
with $t_Ge_i=v_{i-1}$ and $h_Ge_i=v_i$ for $i\in[k]$;
$w$ is {\em closed} if $v_k=v_0$;
if $G$ is a graph,
$w$ is {\em non-backtracking}, or simply {\em NB},
if $\iota_Ge_i\ne e_{i+1}$
for $i\in[k-1]$, and {\em strictly 
non-backtracking closed}, or simply {\em SNBC},
if it is closed, non-backtracking, and 
$\iota_G e_k\ne e_1$.
The {\em visited subgraph} of a walk, $w$, in a graph $G$, denoted
$\ViSu_G(w)$ or simply
$\ViSu(w)$, is the smallest subgraph of $G$ containing all the vertices
and directed edges of $w$;
$\ViSu_G(w)$ generally depends on $G$, i.e., $\ViSu_G(w)$ cannot be inferred
from the sequence $v_0,e_1,\ldots,e_k,v_k$ alone without knowing
$\iota_G$.

The adjacency matrix, $A_G$,
of a graph or digraph, $G$, is defined as usual (its $(v_1,v_2)$-entry
is the number of directed edges from $v_1$ to $v_2$);
if $G$ is a graph on $n$ vertices, 
then $A_G$ is symmetric and we order its eigenvalues (counted with
multiplicities) and denote them
$$
\lambda_1(G)\ge \cdots \ge \lambda_n(G).
$$
If $G$ is a graph, its
Hashimoto matrix (also called the non-backtracking matrix), $H_G$,
is the adjacency matrix of the {\em oriented line graph} of $G$,
$\Line(G)$,
whose vertices are $\Edir_G$ and whose directed edges
are the subset of $\Edir_G\times\Edir_G$ consisting of pairs $(e_1,e_2)$
such that $e_1,e_2$ form the
directed edges of a non-backtracking walk (of length two) in $G$
(the tail of $(e_1,e_2)$ is $e_1$, and its head $e_2$);
therefore $H_G$
is the square matrix indexed on $\Edir_G$, whose $(e_1,e_2)$ entry
is $1$ or $0$ according to, respectively, whether or not
$e_1,e_2$ form a non-backtracking walk
(i.e., $h_G e_1=t_G e_2$ and $\iota_G e_1\ne e_2$).
We use $\mu_1(G)$ to denote the Perron-Frobenius eigenvalue of 
$H_G$, and use $\mu_i(G)$ with $1<i\le \#\Edir_G$ to denote the
other eigenvalues of $H_G$ (which are generally complex-valued)
in any order.

If $B,G$ are both digraphs,
we say that $G$ is a {\em coordinatized graph over $B$
of degree $n$}
if
\begin{equation}\label{eq_coord_def}
V_G=V_B\times [n], \quad\Edir_G=\Edir_B\times[n], \quad
t_G(e,i)=(t_B e,i),\quad
h_G(e,i)=(h_Be,\sigma(e)i)
\end{equation} 
for some map
$\sigma\from\Edir_B\to\cS_n$, where $\cS_n$ is the group
of permutations on $[n]$; we call $\sigma$ (which is uniquely determined by
\eqref{eq_coord_def}) {\em the permutation assignment
associated to $G$}.
[Any such $G$ comes with a map $G\to B$ given by 
``projection to the first component of
the pair,'' and this map is a covering map of degree $n$.]
If $B,G$ are graphs, we say that a graph $G$ is a 
{\em coordinatized graph over $B$
of degree $n$} if \eqref{eq_coord_def} holds and also
\begin{equation}\label{eq_coord_def_graph}
\iota_G(e,i) = \bigl( \iota_B e,\sigma(e)i \bigr) ,
\end{equation} 
which implies that 
\begin{equation}\label{eq_sigma_iota_B}
(e,i)=\iota_G\iota_G(e,i) = \bigl( e, \sigma(\iota_B e)\sigma(e)i \bigr)
\quad\forall e\in\Edir_B,\ i\in[n],
\end{equation}
and hence $\sigma(\iota_B e)=\sigma(e)^{-1}$;
we use $\Coord_n(B)$ to denote the set of all coordinatized covers
of a graph, $B$, of degree $n$.

The {\em order} of a graph, $G$, is $\ord(G)\eqdef (\#E_G)-(\#V_G)$.
Note that a half-loop and a whole-loop each contribute $1$ to 
$\#E_G$ and to the order of $G$.
The {\em Euler characteristic} of a graph, $G$, is
$\chi(G)\eqdef (\# V_G) - (\#\Edir_G)/2$.
Hence $\ord(G)\ge -\chi(G)$, with equality iff $G$ has no half-loops.

If $w$ is a walk in any $G\in\Coord_n(B)$, then one easily
sees that $\ViSu_G(w)$ can be inferred
from $B$ and $w$ alone.

If $B$ is a graph without half-loops, then the {\em permutation model over
$B$} refers to the probability spaces $\{\cC_n(B)\}_{n\in\naturals}$ where
the atoms of $\cC_n(B)$ are coordinatized coverings of degree $n$
over $B$ chosen with the uniform distribution.
More generally, a {\em model} over a graph, $B$, is a collection of
probability spaces, $\{\cC_n(B)\}_{n\in N}$, 
defined for $n\in N$ where $N\subset\naturals$ is an
infinite subset, and where the atoms of each $\cC_n(B)$ are elements
of $\Coord_n(B)$.
There are a number of models related to the permutation model,
which are generalizations of the models of \cite{friedman_alon},
that we call {\em our basic models} and are defined in Article~I;
let us give a rough description.

All of {\em our basic models} are {\em edge independent}, meaning that
for any orientation $\Eor_B\subset\Edir_B$, the values of 
the permutation assignment, $\sigma$, on $\Eor_B$ are independent
of one another (of course, $\sigma(\iota_G e)=(\sigma(e))^{-1}$,
so $\sigma$ is determined by its values on any orientation
$\Eor_B$); for edge independent models, it suffices to specify
the ($\cS_n$-valued)
random variable $\sigma(e)$ for each $e$ in $\Eor_B$ or $\Edir_B$.
The permutation model can be alternatively described as the 
edge independent model that assigns a uniformly chosen permutation
to each $e\in\Edir_B$ (which requires $B$ to have no half-loops);
the {\em full cycle} (or simply {\em cyclic}) model is the same, except
that if $e$ is a whole-loop then $\sigma(e)$ is chosen uniformly
among all permutations whose cyclic structure consists of a single
$n$-cycle.
If $B$ has half-loops, then we restrict $\cC_n(B)$ either to $n$ even
or $n$ odd and for each half-loop $e\in\Edir_B$ we
choose $\sigma(e)$ as follows: if $n$ is even we choose 
$\sigma(e)$ uniformly among all perfect matchings,
i.e., involutions (maps equal to their inverse) with no fixed points;
if $n$ is odd then we choose $\sigma(e)$ uniformly among
all {\em nearly perfect matchings}, meaning involutions with one
fixed point.
We combine terms when $B$ has half-loops: for example,
the term {\em full cycle-involution} (or simply {\em cyclic-involution})
{\em model of odd degree over $B$} refers
to the model where the degree, $n$, is odd,
where $\sigma(e)$ follows the full cycle rule when $e$ is
not a half-loop, and where $\sigma(e)$ is a near perfect matching
when $e$ is a half-loop;
similarly for the {\em full cycle-involution} (or simply 
{\em cyclic-involution})
{\em model of even degree}
and the {\em permutation-involution model of even degree}
or {\em of odd degree}.

If $B$ is a graph, then a model, $\{\cC_n(B)\}_{n\in N}$, over $B$
may well have $N\ne \naturals$ (e.g., our basic models above when
$B$ has half-loops); in this case many formulas involving
the variable $n$ are only defined for $n\in N$.  For brevity, we
often do not explicitly write $n\in N$ in such formulas; 
for example we usually write
$$
\lim_{n\to\infty} \quad\mbox{to abbreviate}\quad
\lim_{n\in N,\ n\to\infty} \ .
$$
Also we often write simply $\cC_n(B)$ or $\{\cC_n(B)\}$ for
$\{\cC_n(B)\}_{n\in N}$ if confusion is unlikely to occur.

A graph is {\em pruned} if all its vertices are of degree at least
two (this differs from the more standard definition of {\em pruned} 
meaning that there are
no leaves).  If $w$ is any SNBC walk in a graph, $G$, then
we easily see that
$\ViSu_G(w)$ is necessarily pruned: i.e., any of its vertices must be
incident upon a whole-loop or two distinct edges
[note that a walk of length $k=1$ about a half-loop, $(v_0,e_1,v_1)$, by
definition, is not SNBC since $\iota_G e_k=e_1$].
It easily follows that $\ViSu_G(w)$ is contained in the graph
obtained from $G$ by repeatedly ``pruning any leaves''
(i.e., discarding any vertex of degree one and its incident edge)
from $G$.
Since our trace methods only concern (Hashimoto matrices and)
SNBC walks, it suffices to work with models $\cC_n(B)$ where
$B$ is pruned.
It is not hard to see that if $B$ is pruned and connected,
then $\ord(B)=0$ iff $B$ is a cycle,
and $\mu_1(B)>1$ iff $\chi(B)<0$;
this is formally proven in Article~III (Lemma~6.4).
Our theorems are not usually interesting unless $\mu_1(B)>\mu_1^{1/2}(B)$,
so we tend to restrict our main theorems
to the case $\mu_1(B)>1$ or, equivalently,
$\chi(B)<0$; some of our techniques work without these restrictions.

\subsection{Asymptotic Expansions}
\label{su_asymptotic_expansions}


A function $f\from\naturals\to\complex$ is a {\em polyexponential} if
it is a sum of functions $p(k)\mu^k$, where $p$ is a polynomial
and $\mu\in\complex$, with the convention
that for $\mu=0$ we understand $p(k)\mu^k$ to mean
any function that vanishes for sufficiently large $k$\footnote{
  This convention is used because then for any fixed matrix, $M$,
  any entry of $M^k$, as a function of $k$, is a polyexponential
  function of $k$; more specifically, the $\mu=0$ convention
  is due to the fact that a Jordan block of eigenvalue $0$ is
  nilpotent.
  }; we refer to the $\mu$
needed to express $f$ as the {\em exponents} or {\em bases} of $f$.
A function $f\from\naturals\to\complex$ is {\em of growth $\rho$}
for a $\rho\in\reals$ if $|f(k)|=o(1)(\rho+\epsilon)^k$ for any $\epsilon>0$.
A function $f\from\naturals\to\complex$ is $(B,\nu)$-bounded if it
is the sum of a function of growth $\nu$ plus a polyexponential function
whose bases are bounded by $\mu_1(B)$ (the Perron-Frobenius eigenvalue
of $H_B$); the {\em larger bases} of $f$ (with respect to $\nu$) are
those bases of the polyexponential function that are larger in
absolute value than $\nu$.
Moreover, such an $f$ is called {\em $(B,\nu)$-Ramanujan} if its
larger bases are all eigenvalues of $H_B$.

We say that a function $f=f(k,n)$ taking some subset of $\naturals^2$ to
$\complex$ has a 
{\em $(B,\nu)$-bounded expansion of order $r$} if for some
constant $C$ we have
\begin{equation}\label{eq_B_nu_defs_summ}
f(k,n) = c_0(k)+\cdots+c_{r-1}(k)+ O(1) c_r(k)/n^r,
\end{equation} 
whenever $f(k,n)$ is defined and $1\le k\le n^{1/2}/C$, where
for $0\le i\le r-1$, the $c_i(k)$ are $(B,\nu)$-bounded and $c_r(k)$
is of growth $\mu_1(B)$.
Furthermore, such an expansion is called {\em $(B,\nu)$-Ramanujan}
if for $0\le i\le r-1$, the $c_i(k)$ are {\em $(B,\nu)$-Ramanujan}.

Typically our functions $f(k,n)$ as in
\eqref{eq_B_nu_defs_summ} are defined for all $k\in\naturals$
and $n\in N$ for an infinite set $N\subset\naturals$ representing
the possible degrees of our random covering maps in the model
$\{\cC_n(B)\}_{n\in N}$ at hand.

\subsection{Tangles}
\label{su_tangles}

A {\em $(\ge\nu)$-tangle} is any 
connected graph, $\psi$, with $\mu_1(\psi)\ge\nu$,
where $\mu_1(\psi)$ denotes the Perron-Frobenius eigenvalue of $H_B$;
a {\em $(\ge\nu,<r)$-tangle} is any $(\ge\nu)$-tangle of order less than
$r$;
similarly for $(>\nu)$-tangles, i.e.,
$\psi$ satisfying the weak inequality $\mu_1(\psi)>\nu$,
and for $(>\nu,r)$-tangles.
We use ${\rm TangleFree}(\ge\nu,<r)$ to denote those graphs that don't
contain a subgraph that is $(\ge\nu,<r)$-tangle, and
${\rm HasTangles}(\ge\nu,<r)$ for those that do; we
never use $(>\nu)$-tangles in defining TangleFree and HasTangles,
for the technical reason
(see Article~III or Lemma~9.2 of \cite{friedman_alon}) that
for $\nu>1$ and any $r\in\naturals$
that there are only finitely many 
$(\ge\nu,<r)$-tangles, up to isomorphism, that are minimal
with respect to inclusion\footnote{
  By contrast, there are infinitely many minimal $(>\nu,<r)$-tangles
  for some values of $\nu>1$ and $r$: indeed, consider any connected pruned
  graph $\psi$, and set $r=\ord(\psi)+2$, $\nu=\mu_1(\psi)$.  Then if
  we fix two vertices in $\psi$ and let $\psi_s$ be the graph that is
  $\psi$ with an additional edge of length $s$ between these two 
  vertices, then $\psi_s$ is an $(>\nu,<r)$-tangle.  However, if
  $\psi'$ is $\psi$ with any single edge deleted, and $\psi'_s$ is 
  $\psi_s$ with this edge deleted, then one can show that
  $\mu_1(\psi'_s)<\nu$ for $s$ sufficiently large.  It follows that
  for $s$ sufficiently large, $\psi_s$ are minimal $(>\nu,<r)$-tangles.
}.

\subsection{$B$-Graphs, Ordered Graphs, and Strongly Algebraic Models}
\label{su_B_ordered_strongly_alg}

An {\em ordered graph}, $G^\og$, is a graph, $G$, endowed with an
{\em ordering}, meaning
an orientation (i.e., $\iota_G$-orbit representatives), 
$\Eor_G\subset\Edir_G$, 
and total orderings of $V_G$ and $E_G$;
a walk, $w=(v_0,\ldots,e_k,v_k)$ in a graph endows $\ViSu(w)$ with a
{\em first-encountered} ordering:
namely, $v\le v'$ if the first occurrence of $v$ comes before that
of $v'$ in the sequence $v_0,v_1,\ldots,v_k$,
similarly for $e\le e'$, and we orient each edge in the
order in which it is first traversed (some edges may be traversed
in only one direction).
We use $\ViSu^\og(w)$ to refer to $\ViSu(w)$ with this ordering.

A {\em morphism} $G^\og\to H^\og$ of ordered graphs is a morphism
$G\to H$ that respects the ordering in the evident fashion.
We are mostly interested in {\em isomorphisms} of ordered graphs;
we easily see that any isomorphism $G^\og\to G^\og$ must be the
identity morphism; it follows that if $G^\og$ and $H^\og$ are
isomorphic, then there is a unique isomorphism $G^\og\to H^\og$.

If $B$ is a graph, then a $B$-graph, $G_\Bg$, is a graph $G$ endowed 
with a map $G\to B$ (its {\em $B$-graph} structure).
A {\em morphism} $G_\Bg\to H_\Bg$ of $B$-graphs is a morphism
$G\to H$ that respects the $B$-structures in the evident sense.
An {\em ordered $B$-graph}, $G^\og_\Bg$, is a graph endowed with
both an ordering and a $B$-graph structure; a morphism of
ordered $B$-graphs is a morphism of the underlying graphs that
respects both the ordering and $B$-graph structures.
If $w$ is a walk in a $B$-graph, $G_\Bg$, we use $\ViSu_\Bg(w)$ to denote
$\ViSu(w)$ with the $B$-graph structure it inherits from $G$ in
the evident sense; we use $\ViSu_\Bg^\og(w)$ to denote
$\ViSu_\Bg(w)$ with its first-encountered ordering.

At times we drop the superscript $\,^\og$ and the subscript $\,_\Bg$;
for example, we write $G\in\Coord_n(B)$ instead of $G_\Bg\in\cC_n(B)$
(despite the fact that we constantly utilize
the $B$-graph structure on elements of
$\Coord_n(B)$).

A $B$-graph $G_\Bg$ is {\em covering} or {\'etale} if its structure
map $G\to B$ is.

If $\pi\from S\to B$ is a $B$-graph, we use
$\mec a=\mec a_{S_\Bg}$ to denote the vector
$\Edir_B\to\integers_{\ge 0}$ given by
$a_{S_\Bg}(e) = \# \pi^{-1}(e)$;
since $a_{S_\Bg}(\iota_B e) = a_{S_\Bg}(e)$ for all $e\in\Edir_B$,
we sometimes view $\mec a$ as a function $E_B\to\integers_{\ge 0}$, i.e.,
as the function taking $\{e,\iota_B e\}$ to 
$a_{S_\Bg}(e)=a_{S_\Bg}(\iota_B e)$.
We similarly define $\mec b_{S_\Bg}\from V_B\to\integers_{\ge 0}$ by
setting $b_{S_\Bg}(v) = \#\pi^{-1}(v)$.
If $w$ is a walk in a $B$-graph, we set $\mec a_w$ to be
$\mec a_{S_\Bg}$ where $S_\Bg=\ViSu_\Bg(w)$, and similarly for $\mec b_w$.
We refer to $\mec a,\mec b$ (in either context) as
{\em $B$-fibre counting functions}.

If $S_\Bg^\og$ is an ordered $B$-graph and $G_\Bg$ is a $B$-graph, we 
use $[S_\Bg^\og]\cap G_\Bg$ to denote the set of ordered graphs ${G'}_\Bg^\og$
such that $G'_\Bg\subset G_\Bg$ and ${G'}_\Bg^\og\isom S_\Bg^\og$
(as ordered $B$-graphs); this set is naturally identified with the
set of injective morphisms $S_\Bg\to G_\Bg$, and the cardinality of these
sets is independent of the ordering on $S_\Bg^\og$.


A $B$-graph, $S_\Bg$, or an ordered $B$-graph, $S_\Bg^\og$,
{\em occurs in a model $\{\cC_n(B)\}_{n\in N}$}
if for all sufficiently large
$n\in N$, $S_\Bg$ is isomorphic to a $B$-subgraph of some element
of $\cC_n(B)$; similary a graph, $S$, {\em occurs in 
$\{\cC_n(B)\}_{n\in N}$} if it can be endowed with a $B$-graph
structure, $S_\Bg$, that occurs in 
$\{\cC_n(B)\}_{n\in N}$.

A model $\{\cC_n(B)\}_{n\in N}$ of coverings of $B$ is {\em strongly
algebraic} if
\begin{enumerate}
\item for each $r\in\naturals$
there is a function, $g=g(k)$, of growth $\mu_1(B)$
such that if $k\le n/4$ we have
\begin{equation}\label{eq_algebraic_order_bound}
\EE_{G\in\cC_n(B)}[ \snbc_{\ge r}(G,k)] \le
g(k)/n^r
\end{equation}
where $\snbc_{\ge r}(G,k)$ is the number of SNBC walks of length
$k$ in $G$ whose visited subgraph is of order at least $r$;
\item
for any $r$ there exists
a function $g$ of growth $1$ and real $C>0$ such that the following
holds:
for any ordered $B$-graph, $S_\Bg^\og$, that is pruned and of
order less than $r$,
\begin{enumerate}
\item
if $S_\Bg$ occurs in $\cC_n(B)$, then for
$1\le \#\Edir_S\le n^{1/2}/C$,
\begin{equation}\label{eq_expansion_S}
\EE_{G\in\cC_n(B)}\Bigl[ \#\bigl([S_\Bg^\og]\cap G\bigr) \Bigr]
=
c_0 + \cdots + c_{r-1}/n^{r-1}
+ O(1) g(\# E_S) /n^r
\end{equation} 
where the $O(1)$ term is bounded in absolute value by $C$
(and therefore independent of $n$ and $S_\Bg$), and
where $c_i=c_i(S_\Bg)\in\reals$ such that
$c_i$ is $0$ if $i<\ord(S)$ and $c_i>0$ for $i=\ord(S)$;
and
\item
if $S_\Bg$ does not occur in $\cC_n(B)$, then for any
$n$ with $\#\Edir_S\le n^{1/2}/C$,
\begin{equation}\label{eq_zero_S_in_G}
\EE_{G\in\cC_n(B)}\Bigl[ \#\bigl([S_\Bg^\og]\cap G\bigr) \Bigr]
= 0 
\end{equation} 
(or, equivalently, no graph in $\cC_n(B)$ has a $B$-subgraph isomorphic to
$S_\Bg^\og$);
\end{enumerate}
\item
$c_0=c_0(S_\Bg)$ equals $1$ if $S$ is a cycle (i.e., $\ord(S)=0$ and
$S$ is connected) that occurs in $\cC_n(B)$;
\item
$S_\Bg$ occurs in $\cC_n(B)$ iff $S_\Bg$ is an \'etale $B$-graph
and $S$ has no half-loops; and
\item
there exist
polynomials $p_i=p_i(\mec a,\mec b)$ such that $p_0=1$
(i.e., identically 1), and for every
\'etale $B$-graph, $S_\Bg^\og$, we have that
\begin{equation}\label{eq_strongly_algebraic}
c_{\ord(S)+i}(S_\Bg) = p_i(\mec a_{S_\Bg},\mec b_{S_\Bg}) \ .
\end{equation}
\end{enumerate}
Notice that condition~(3), regarding $S$ that are cycles, is implied
by conditions~(4) and~(5); we leave in condition~(3) since this makes the
definition of {\em algebraic} (below) simpler.
Notice that \eqref{eq_expansion_S} and \eqref{eq_strongly_algebraic}
are the main reasons that we work with
ordered $B$-graphs: indeed, the coefficients depend only on
the $B$-fibre counting function $\mec a,\mec b$, which 
depend on the structure of
$S_\Bg^\og$ as a $B$-graph; this is not true if we don't work with
ordered graphs: i.e.,
\eqref{eq_expansion_S} fails to
hold if we replace $[S_\Bg^\og]$
with $[S_\Bg]$ (when $S_\Bg$ has nontrivial automorphisms), where
$[S_\Bg]\cap G$ refers to the number of $B$-subgraphs of $G$ isomorphic
to $S_\Bg$; the reason is that
$$
\#[S_\Bg^\og]\cap G_\Bg = \bigl( \#{\rm Aut}(S_\Bg)\bigr)
\bigl( \#[S_\Bg]\cap G_\Bg \bigr)
$$
where ${\rm Aut}(S_\Bg)$ is the group of automorphisms of $S_\Bg$, 
and it is $[S_\Bg^\og]\cap G_\Bg$ rather than $[S_\Bg]\cap G_\Bg$
that turns out to have the ``better'' properties;
see Section~6 of Article~I for examples.
Ordered graphs are convenient to use for a number of other reasons.

\ignore{
\myDeleteNote{Stuff deleted here and below on September 13, 2018.}
}

\subsection{Homotopy Type}

The homotopy type of a walk and of an ordered subgraph are defined
by {\em suppressing} its ``uninteresting'' vertices of degree two;
examples are given in Section~6 of Article~I.
Here is how we make this precise.

A {\em bead} in a graph is a vertex of degree two that is not
incident upon a self-loop.
Let $S$ be a graph and $V'\subset V_S$ be a {\em proper bead subset} of 
$V_S$,
meaning that $V'$ consists only of beads of $V$,
and that no connected component of $S$ has all its vertices in
$V'$ (this can only happen for connected components of $S$ that
are cycles);
we define the {\em bead suppression} $S/V'$ to be the following
graph: (1) its
vertex set $V_{S/V'}$
is $V''=V_S\setminus V'$, (2) its directed edges, $\Edir_{S/V'}$ consist
of
the {\em $V$'-beaded paths}, i.e., non-backtracking walks
in $S$ between elements of $V''$ whose intermediate vertices lie in $V'$,
(3) $t_{S/V'}$ and $h_{S/V'}$ give the first and last vertex of
the beaded path, and (4) $\iota_{S/V'}$ takes a beaded path
to its reverse walk
(i.e., takes $(v_0,e_1,\ldots,v_k)$ to
$(v_k,\iota_S e_k,\ldots,\iota_S e_1,v_0)$).
One can recover $S$ from the suppression $S/V'$ for pedantic reasons,
since we have defined its directed edges to be beaded paths of $S$.
If $S^\og=\ViSu^\og(w)$ where $w$ is a non-backtracking walk,
then the ordering of $S$ can be inferred by the naturally
corresponding order on $S/V'$, and we use $S^\og/V'$ to denote
$S/V'$ with this ordering.

Let $w$ be a non-backtracking walk in a graph, and 
$S^\og=\ViSu^\og(w)$ its visited
subgraph; the {\em reduction} of $w$ is the ordered graph,
$R^\og$, denoted $S^\og/V'$,
whose underlying graph is
$S/V'$ where $V'$ is the set of beads of $S$ except
the first and last vertices of $w$ (if one or both are beads),
and whose ordering is naturally arises from that on $S^\og$;
the {\em edge lengths} of $w$ is the function $E_{S/V'}\to\naturals$
taking an edge of $S/V'$ to the length of the beaded path it represents
in $S$;
we say that $w$ is {\em of homotopy type} $T^\og$ for any ordered
graph $T^\og$ that is isomorphic to $S^\og/V'$; in this case
the lengths of $S^\og/V'$ naturally give lengths $E_T\to\naturals$
by the unique isomorphism from $T^\og$ to $S^\og/V'$.
If $S^\og$ is the visited subgraph of a non-backtracking walk,
we define the reduction, homotopy type, and edge-lengths of $S^\og$ to
be that of the walk, since these notions depend only on $S^\og$ and
not the particular walk.

If $T$ is a graph and $\mec k\from E_T\to\naturals$ a function, then
we use $\VLG(T,\mec k)$ (for {\em variable-length graph}) to denote
any graph obtained from $T$ by gluing in a path of length $k(e)$
for each $e\in E_T$.  If $S^\og$ is of homotopy type $T^\og$
and $\mec k\from E_T\to \naturals$ its edge lengths,
then $\VLG(T,\mec k)$ is isomorphic to $S$ (as a graph).
Hence the construction of variable-length graphs is a sort of
inverse to bead suppression.

If $T^\og$ is an ordering on $T$ that arises as the first encountered
ordering of a non-backtracking walk on $T$ (whose visited subgraph
is all of $T$), then this ordering gives rise to a natural
ordering on $\VLG(T,\mec k)$ that we denote $\VLG^\og(T^\og,\mec k)$.
Again, this ordering on the variable-length graph is a sort of
inverse to bead suppression on ordered graphs.

\subsection{$B$-graphs and Wordings}

If $w_B=(v_0,e_1,\ldots,e_k,v_k)$ with $k\ge 1$ is a walk in a graph
$B$, then we can identify
$w_B$ with the string $e_1,e_2,\ldots,e_k$ over the alphabet
$\Edir_B$.
For technical reasons, the definitions below of
a {\em $B$-wording} and 
the {\em induced wording}, are given as strings over $\Edir_B$ rather
than the full alternating string of vertices and directed edges.
The reason is that 
doing this gives the correct notion of the {\em eigenvalues} of
an algebraic model (defined below).

Let $w$ be a non-backtracking walk in a $B$-graph, whose reduction
is $S^\og/V'$, and let
$S_\Bg^\og=\ViSu_\Bg^\og$.
Then the {\em wording induced by $w$} on $S^\og/V'$ is
the map $W$ from $\Edir_{S/V'}$ to strings in $\Edir_B$
of positive length, 
taking a
directed edge $e\in\Edir_{S/V'}$ to the string of $\Edir_B$ edges
in the non-backtracking walk in $B$
that lies under the walk in $S$ that it represents.
Abstractly, we say that a {\em $B$-wording} of a graph $T$
is a map $W$ from $\Edir_T$ to words over the alphabet
$\Edir_B$ that represent (the directed edges of)
non-backtracking walks in $B$ such that
(1) $W(\iota_T e)$ is the reverse word (corresponding to
the reverse walk) in $B$ of $W(e)$, 
(2) if $e\in\Edir_T$ is a half-loop, then $W(e)$ is of length one
whose single letter is a half-loop, and
(3) the tail of the first directed edge in $W(e)$ 
(corresponding to the first vertex in the associated walk in $B$)
depends only on $t_T e$;
the {\em edge-lengths} of $W$ is the function $E_T\to\naturals$
taking $e$ to the length of $W(e)$.
[Hence the wording induced by $w$ above is, indeed, a $B$-wording.]

Given a graph, $T$, and a $B$-wording $W$, there is a $B$-graph,
unique up to isomorphism, whose underlying graph is $\VLG(T,\mec k)$
where $\mec k$ is the edge-lengths of $W$, and where the $B$-graph
structure maps the non-backtracking walk in $\VLG(T,\mec k)$
corresponding to an $e\in\Edir_T$ to the non-backtracking walk in $B$
given by $W(e)$.
We denote any such $B$-graph by $\VLG(T,W)$; again this is
a sort of inverse to starting with a non-backtracking walk
and producing the wording it induces on its visited subgraph.

Notice that if $S_\Bg^\og=\VLG(T^\og,W)$ for a $B$-wording, $W$,
then the $B$-fibre counting functions
$\mec a_{S_\Bg}$ and $\mec b_{S_\Bg}$ can be
inferred from $W$, and we may therefore write $\mec a_W$ and
$\mec b_W$.

\subsection{Algebraic Models}

By a $B$-type we mean a pair $T^{\rm type}=(T,\cR)$ consisting
of a graph, $T$, and a map from $\Edir_T$ to the set
of regular languages over the alphabet $\Edir_B$ (in the sense of regular
language theory) such that
(1) all words in $\cR(e)$ are positive length strings corresponding to
non-backtracking walks in $B$, 
(2) if for $e\in\Edir_T$ we have $w=e_1\ldots e_k\in\cR(e)$,
then $w^R\eqdef \iota_B e_k\ldots\iota_B e_1$ lies in $\cR(\iota_T e)$,
and (3) if $W\from \Edir_T\to(\Edir_B)^*$ (where $(\Edir_B)^*$ is
the set of strings over $\Edir_B$) satisfies
$W(e)\in\cR(e)$ and $W(\iota_T e)=W(e)^R$ for all $e\in \Edir_T$,
then $W$ is a $B$-wording.
A $B$-wording $W$ of $T$ is {\em of type $T^{\rm type}$} if
$W(e)\in\cR(e)$ for each $e\in\Edir_T$.

Let $\cC_n(B)$ be a model that satisfies (1)--(3) of the definition
of strongly algebraic.
If $\cT$ a subset of $B$-graphs,
we say that the model is {\em algebraic restricted to $\cT$}
if 
either all $S_\Bg\in\cT$ occur in $\cC_n(B)$ or they all do not,
and if so
there are polynomials $p_0,p_1,\ldots$ such that
$c_i(S_\Bg)=p_i(S_\Bg)$ for any $S_\Bg\in\cT$. 
We say that $\cC_n(B)$ is {\em algebraic} if 
\begin{enumerate}
\item
setting $h(k)$ to be
the number of $B$-graph isomorphism classes of \'etale $B$-graphs
$S_\Bg$ such that $S$ is a cycle of length $k$ and $S$ does
not occur in $\cC_n(B)$, we have that 
$h$ is a function of growth $(d-1)^{1/2}$; and
\item
for any
pruned, ordered graph, $T^\og$, there is a finite number of
$B$-types, $T_j^{\rm type}=(T^\og,\cR_j)$, $j=1,\ldots,s$, 
such that (1) any $B$-wording, $W$, of $T$ belongs to exactly one
$\cR_j$, and
(2) $\cC_n(B)$ is algebraic when restricted to $T_j^{\rm type}$.
\end{enumerate}

[In Article~I we show that
if instead each $B$-wording belong to 
{\em at least one} $B$-type $T_j^{\rm type}$, then one can choose a
another set of
$B$-types that satisfy (2) and where each $B$-wording belongs
to {\em a unique} $B$-type;
however, the uniqueness
is ultimately needed in our proofs,
so we use uniqueness in our definition of algebraic.]

We remark that one can say that a walk, $w$, in a $B$-graph,
or an ordered $B$-graphs, $S_\Bg^\og$, is of {\em homotopy type $T^\og$},
but when $T$ has non-trivial automorphism one {\em cannot} say
that is of $B$-type $(T,\cR)$ unless---for example---one orders
$T$ and speaks of an {\em ordered $B$-type}, $(T^\og,\cR)$.
[This will be of concern only in Article~II.]

We define the {\em eigenvalues} of a regular language, $R$, to be the minimal
set $\mu_1,\ldots,\mu_m$ such that for any $k\ge 1$,
the number of words of length $k$ in the language
is given as
$$
\sum_{i=1}^m p_i(k)\mu_i^k
$$
for some polynomials $p_i=p_i(k)$, with the convention that
if $\mu_i=0$ then $p_i(k)\mu_i^k$ refers to any function that 
vanishes for $k$ sufficiently large (the reason for this is that
a Jordan block of eigenvalue $0$ is a nilpotent matrix).
Similarly, we define the eigenvalues of a $B$-type $T^{\rm type}=(T,\cR)$
as the union of all the eigenvalues of the $\cR(e)$.
Similarly a {\em set of eigenvalues} of a graph, $T$
(respectively, an algebraic model, $\cC_n(B)$)
is
any set containing the eigenvalues containing the eigenvalues
of some choice of $B$-types used in the definition of algebraic
for $T$-wordings (respectively, for $T$-wordings for all $T$).

[In Article~V we prove that all of our basic models are algebraic;
some of our basic models, such as the
permutation-involution model and the cyclic models, are not
strongly algebraic.]

We remark that a homotopy type, $T^\og$,
of a non-backtracking walk, can only have beads as its first or last 
vertices; however, in the definition of algebraic we require
a condition on {\em all pruned graphs}, $T$, 
which includes $T$ that may have many beads and may not be connected;
this is needed
when we define homotopy types of pairs in Article~II.

\subsection{SNBC Counting Functions}

If $T^\og$ is an ordered graph and $\mec k\from E_T\to\naturals$, 
we use $\SNBC(T^\og,\mec k;G,k)$ to denote the set of SNBC walks in $G$
of length $k$ and of homotopy type $T^\og$ and edge lengths $\mec k$.
We similarly define
$$
\SNBC(T^\og,\ge\bec\xi;G,k) \eqdef 
\bigcup_{\mec k\ge\bec\xi} \SNBC(T^\og,\mec k;G,k)
$$
where $\mec k\ge\bec\xi$ means that $k(e)\ge\xi(e)$ for all $e\in E_T$.
We denote the cardinality of these sets by replacing $\SNBC$ with
$\snbc$;
we call $\snbc(T^\og,\ge\bec\xi;G,k)$ the set of 
{\em $\bec\xi$-certified
traces of homotopy type $T^\og$ of length $k$ in $G$};
in Article~III we will refer to certain $\bec\xi$ as {\em certificates}.

\section{The Main Theorems in this Article}
\label{se_main_Ramanujan}

In this section we formally state the main theorems in this
article.  We first review some definitions and results 
of Article~V.

\subsection{Results from Article~V}

If $B$ is a graph, $\|A_{\widehat B}\|_2$ denotes
the $L^2$ norm of the adjacency operator
on a universal cover, $\widehat B$, of $B$; it is well-known that
if $B$ is $d$-regular,
then $\|A_{\widehat B}\|_2=2\sqrt{d-1}$
\cite{mohar_woess}.
If $\pi\from G\to B$ is a covering map graphs,
and $\epsilon>0$, the
{\em $\epsilon$-non-Alon multiplicity of $G$
relative to $B$} is
$$
{\rm NonAlon}_B(G;\epsilon) \eqdef
\# \bigl\{\lambda\in\specnew_B(A_G)\ \bigm|\
|\lambda|>
\|A_{\widehat B}\|_2 +\epsilon\bigr\} ,
$$
where
the above $\lambda$ are counted with their multiplicity in
$\specnew_B(A_G)$.

In Article~V the Relativized Alon Conjecture was proven when
$B$ is $d$-regular.  The statement regards
any 
algebraic model $\{\cC_n(B)\}_{n\in N}$ an algebraic
model over a $d$-regular graph $B$;
it says that for
$\epsilon>0$ there is a constant $C=C(\epsilon)$ for which
$$
\Prob_{G\in\cC_n(B)}[ {\rm NonAlon}_B(G;\epsilon)>0 ]
\le  C(\epsilon)/n \ .
$$
The point of this article is to give matching upper and lower bounds
for this probability when $B$ is, furthermore, a
{\em Ramanujan graph} in the following sense.

\begin{definition}\label{de_Ramanujan}
We say that a $d$-regular graph $B$ is {\em Ramanujan} if all eigenvalues
of $A_B$ lie in
$$
\{d,-d\} \cup \Bigl[ -2\sqrt{d-1}, 2\sqrt{d-1} \Bigr] .
$$
\end{definition}

We now give the more precise form of the Relativized Alon
Conjecture proven in Article~V.

\begin{definition}\label{de_tau_tang}
Let $\{\cC_n(B)\}_{n\in N}$ be a model over a graph, $B$.
By the {\em tangle power of $\{\cC_n(B)\}$}, denoted $\tau_{\rm tang}$,
we mean the smallest order, $\ord(S)$, of any graph, $S$, that
occurs in $\{\cC_n(B)\}$ and satisfies
$\mu_1(S)>\mu_1^{1/2}(B)$.
\end{definition}

In this article we prove some results regarding $\tau_{\rm tang}$;
for example, 
the results of Section~6.3 of \cite{friedman_alon}
show that for any algebraic model over a $d$-regular graph, $B$,
$$
\tau_{\rm tang} \ge  m=m(d)
$$
where
$$
m(d) =
\Bigl\lfloor \bigl( (d-1)^{1/2} - 1 \bigr)/2  \Bigr\rfloor +1 
$$
(and for any $d\ge 3$ there is a $d$-regular $B$ where
equality holds).

The most difficult theorem in this series of articles, to which most
of Articles~II-V are devoted, is the following result.

\begin{theorem}\label{th_rel_Alon_regular2}
Let $\cC_n(B)$ be an
algebraic model
over a $d$-regular graph $B$.
For any $\nu$ with
$(d-1)^{1/2}<\nu<d-1$, let $\epsilon'>0$ be given by
$$
2(d-1)^{1/2}+ \epsilon' = \nu + \frac{d-1}{\nu}.
$$
Then
\begin{enumerate}
\item
there is an
integer $\tau=\tau_{\rm alg}(\nu,r)\ge 1$ such that
for any sufficiently small $\epsilon>0$ there are constants
$C=C(\epsilon),C'>0$ such that for sufficiently large $n$ we have
\begin{equation}\label{eq_rel_alon_expect_lower_and_upper}
n^{-\tau} C'
\le
\EE_{G\in\cC_n(B)}[ \II_{{\rm TangleFree}(\ge\nu,<r)}(G)
{\rm NonAlon}_d(G;\epsilon'+\epsilon) ]
\le
n^{-\tau}   C(\epsilon) ,
\end{equation}
or
\item
for all $j\in\naturals$ and $\epsilon>0$ we have
\begin{equation} 
\label{eq_rel_alon_expect_upper_infinite}
\EE_{G\in\cC_n(B)}[ \II_{{\rm TangleFree}(\ge\nu,<r)}(G)
{\rm NonAlon}_d(G;\epsilon'+\epsilon) ] \le O(n^{-j})
\end{equation} 
in which case we use the notation $\tau_{\rm alg}(\nu,r)=+\infty$.
\end{enumerate}
Moreover, if $\tau=\tau_{\rm alg}(\nu,r)$ is finite, then for
some eigenvalue, $\ell\in\reals$, of the model with $|\ell|>\nu$,
there is a real $C_\ell>0$
such that for sufficiently small $\theta>0$
\begin{equation}\label{eq_new_eigenvalues_near_ell}
\lim_{n\to\infty} 
\EE_{G\in\cC_n(B)}\bigl[\#
\bigl(\specnew_B(H_G)\cap B_{n^{-\theta}}(\ell) \bigr) 
\II_{{\rm TangleFree}(\ge\nu,<r)}(G)
\bigr]
= C_\ell n^{-\tau} + o(n^{-\tau}) .
\end{equation} 
\end{theorem}

%
Notice if $\nu_1\le\nu_2$ and $r_1\ge r_2$ then
$$
\II_{{\rm TangleFree}(\ge\nu_2,<r_2)}(G)  
\le \II_{{\rm TangleFree}(\ge\nu_1,<r_1)}(G) ,
$$
for the simple reason that 
$\II_{{\rm TangleFree}(\ge\nu_2,<r_2)}(G)=1$ implies that
$G$ has no $(\ge\nu_2,<r_2)$-tangles, and hence no
$(\ge\nu_1,<r_1)$-tangles;
then \eqref{eq_rel_alon_expect_lower_and_upper} and
\eqref{eq_rel_alon_expect_upper_infinite} imply that
\begin{equation}\label{eq_tau_alg_nu_r_compare}
\tau_{\rm alg}(\nu_1,r_1) \le
\tau_{\rm alg}(\nu_2,r_2).
\end{equation}

\begin{definition}\label{de_algebraic_power}
Let $\{\cC_n(B)\}_{n\in N}$ be an algebraic model over a $d$-regular
graph $B$.
For each $r\in\naturals$ and $\nu$ with $(d-1)^{1/2}<\nu<d-1$, let
$\tau(\nu,r)$ be as in Theorem~\ref{th_rel_Alon_regular2}.
We define the {\em algebraic power} of the model $\cC_n(B)$ to be
$$
\tau_{\rm alg} =
\max_{\nu>(d-1)^{1/2},r} \tau(\nu,r) =
\limsup_{r\to\infty,\ \nu\to(d-1)^{1/2}}
\tau(\nu,r)
$$
where $\nu$ tends to $(d-1)^{1/2}$ from above
(and we allow $\tau_{\rm alg}=+\infty$ when this maximum is
unbounded or if $\tau(\nu,r)=\infty$ for some $r$ and $\nu>(d-1)^{1/2}$).
\end{definition}
Of course, according to Theorem~\ref{th_rel_Alon_regular2},
$\tau(\nu,r)\ge 1$ for all $r$ and all relevant $\nu$, and hence
$\tau_{\rm alg}\ge 1$.

Here is the more precise form of the Relativized Alon
Conjecture proven in Article~V.

\begin{theorem}\label{th_rel_Alon_regular}
Let $B$ be a $d$-regular graph, and let
$\cC_n(B)$ be an
algebraic model of tangle power $\tau_{\rm tang}$
and algebraic power $\tau_{\rm alg}$ (both of which are at least $1$).
Let
$$
\tau_1 = \min(\tau_{\rm tang},\tau_{\rm alg}), \quad
\tau_2 = \min(\tau_{\rm tang},\tau_{\rm alg}+1).
$$
Then $\tau_2\ge \tau_1\ge 1$, and
for $\epsilon>0$ sufficiently small there are
$C,C'$ such that for sufficiently large $n$ we have
\begin{equation}\label{eq_first_main}
C' n^{-\tau_2}
\le
\Prob_{G\in\cC_n(B)}\bigl[
{\rm NonAlon}_d(G;\epsilon)>0
\bigr]
\le
C n^{-\tau_1}.
\end{equation}
\end{theorem}

The last result we need from Article~V regards a set of eigenvalues
for our basic models.

\begin{lemma}
Let $B$ be a connected, pruned graph with $\mu_1(B)>1$
(equivalently $\chi(B)<0$).  All our basic models are algebraic,
and a set of eigenvalues for each model
consist of possibly $1$ and some subset of
the eigenvalues $\mu_i(B)$ of the Hashimoto matrix $H_B$.
\end{lemma}
[The Ihara determinantal formula (see Articles~I or V)
easily implies that all $B$'s in the above lemma have
at least one $H_B$ eigenvalue
equal to either $\pm 1$; hence the possible addition of $1$ to
the set of eigenvalues in the lemma is 
not particularly significant.]

\subsection{Main Result of This Article}

In principle we can compute $\tau_{\rm alg}$, using
the methods of Articles~II-V, which involve analyzing the main term
of certain asymptotic expansions involving certified traces.
However this computation is difficult to carry out.
We will borrow the method of \cite{friedman_random_graphs}
that uses the existence of these asymptotic expansion 
and an indirect method to draw conclusions about the main terms
we need.

\begin{theorem}\label{th_second_main_theorem}
Let $d\ge 3$ be an integer,
and let $\{\cC_n(B)\}_{n\in N}$ be one of our basic models
over $d$-regular Ramanujan
graph, $B$.
Then $\tau_{\rm alg}=+\infty$.
\end{theorem}

The idea behind the proof is to show that
\eqref{eq_new_eigenvalues_near_ell}
cannot hold for any fixed value of $\tau$
with $r\to\infty$ if $\ell=d-1$, due to the fact that a new eigenvalue
of $H_G$ near $d-1$ implies that $G$ has a ``nearly disconnected
component,'' a notion which is made precise by Alon's notion
of {\em magnification}
(one could also use an analog of ``Cheeger's'' inequality for
graphs, e.g., \cite{dodziuk,jerrum1}).
To prove this one needs to prove a (fairly weak)
magnification result for most graphs in the model $\cC_n(B)$.
This result holds for all of our basic models.

If $B$ is $d$-regular Ramanujan and connected, then the larger $H_B$
eigenvalues of all our basic models are either $d-1$ or
$\pm (d-1)$ (the latter iff $B$ is bipartite), and we easily see
that if \eqref{eq_new_eigenvalues_near_ell}
for some $\ell$ then it must hold for $\ell=d-1$.
Since this is impossible, we must have
$\tau_{\rm alg}=+\infty$.

\subsection{Results on $\tau_{\rm tang}$}

Whenever $\tau_{\rm alg}=+\infty$, or merely
$\tau_{\rm alg}\ge \tau_{\rm tang}+1$,
Theorem~\ref{th_rel_Alon_regular}
determines matching upper and lower bounds on
$$
\Prob_{G\in\cC_n(B)}\bigl[
{\rm NonAlon}_d(G;\epsilon)>0
\bigr]
$$
for any fixed $\epsilon>0$ sufficiently small, both bounds being 
proportional to $n^{-\tau_{\rm tang}}$.
It therefore becomes interesting to compute $\tau_{\rm tang}$ or
to give bounds on it.

In Section~\ref{se_tau_tangle} we shall give such bounds on $\tau_{\rm tang}$.
Let us state the main bounds.

If we fix $d\ge 3$, then the lower bound we give on $\tau_{\rm tang}$
for any $d$-regular $B$ is
$$
\tau_{\rm tang} \ge \Bigl\lfloor \bigl( (d-1)^{1/2}+1 \bigr)/2 \Bigr\rfloor
$$
where $\lfloor\ \rfloor$ denotes the floor function, i.e., the largest
integer lower bound;
this bound is tight when $B$ is a bouquet of $d/2$ whole-loops
(so that $d$ is even) and $\cC_n(B)$ is the permutation model.
Furthermore, in models over $B$ in which whole-loops don't occur, we have
$$
\tau_{\rm tang} \ge \bigl\lfloor (d-1)^{1/2} \bigr\rfloor;
$$
this bound is tight in our basic model
whenever $B$ is a bouquet of $d$ half-loops,
and is also tight, except for possibly $d=4$,
for the full cycle model of $d/2$ whole-loops (hence $d$ is even).

As noted in \cite{friedman_alon}, this implies that the full cycle
model has a much lower probability of having non-Alon new eigenvalues
than does the permutation model, at least when $B$ is a bouquets of 
sufficiently many whole-loops.

We also prove that for fixed $d$, as the girth of $B$ tends to
infinity, then so does $\tau_{\rm tang}$.
Hence the lower bounds quoted above can be very
far from tight.
Our proof, however, does not give an explicit relationship between
the girth and $\tau_{\rm tang}$.

\section{Magnifiers and Tangles}
\label{se_new_magnifiers}

In this section we describe some technical results we will prove about the
relative {\em magnification} of random graphs in our basic models.
One could alternatively use a graph theoretic analog
\cite{dodziuk,jerrum1,jerrum2} of
``Cheeger's'' inequality \cite{cheeger};
in this article we will use
magnification.


\subsection{Magnifiers}

We review the results of Alon on {\em magnifiers}.

\begin{definition}
Let $G$ be a graph, and $U\subset V_G$.  We define the
{\em neighbourhood of $U$}, denoted $\Gamma_G(U)$, to be the subset
of $V_G$ consisting of those vertices joined by
an edge of $G$ to a vertex of $U$.
If $\gamma>0$ is a real number,
we say that a graph, $G$, is a {\em $\gamma$-magnifier}
if for all $U\subset V_G$ of size at most $(\#V_G)/2$ we have
$$
\#\bigl( \Gamma_G(U)\setminus U \bigr) \ge \gamma(\#U)  \ ;
$$
moreover, we say that $G$ is a {\em $\gamma$-spreader}
if for all such $U$ we have
$$
\#\bigl(\Gamma_H(U)\bigr) \ge (1+\gamma)(\# U).
$$
\end{definition}

The notion of a magnifier was introduced in \cite{alon_eigenvalues},
where Alon proved the following theorem.
\begin{theorem}[Alon, \cite{alon_eigenvalues}]
If $G$ is $d$-regular and a $\gamma$-magnifier, then for all $i>1$ we have
$$
\lambda_i(G) \le d - \frac{\gamma^2}{4+2\gamma^2}.
$$
\end{theorem}

The notion of a spreader appears in
\cite{friedman_random_graphs,friedman_alon} but in this article
we will only use the notion of magnification; the point is that
it is easier to prove that most random $d$-regular
graphs on $n$ vertices are $\gamma$-spreaders, since spreading is a less
subtle feature than magnification.
However, in this article (unlike \cite{friedman_alon})
a graph $G\in\cC_n(B)$ will never be a spreader if $B$ is a connected,
bipartite graph
(since the subset of all
vertices lying over one side of a bipartition
$B$ has all its neighbours in the other side).

\subsection{Pseudo-Magnification}

In article we will study the following variant of magnification.

\begin{definition}\label{de_pseudo_magnification}
For real $\gamma>0$ and $R\in\naturals$, we say that a graph,
$G$ is an $(R,\gamma)$-pseudo-magnifier if for each
$U\subset V_G$ with
$$
R \le \# U \le (\#V_G)/2
$$
we have
$$
\#\bigl( \Gamma_G(U)\setminus U \bigr) \ge \gamma(\#U) .
$$
\end{definition}

Our interest in this definition is evident in the following definition
and easy lemma.

\begin{definition}
Let $\cC_n(B)$ be a model over a connected graph, $B$.  We say that
the model is {\em pseudo-magnifying} if 
for every $i\in\naturals$ there is a $\gamma>0$ and
$R\in\naturals$ for which
$$
\Prob_{G\in\cC_n(B)}[\mbox{$G$ is not an
$(R,\gamma)$-pseudo-magnifier}] \le O(n^{-i}).
$$
\end{definition}

\begin{lemma}\label{le_no_coeff_largest_base}
Let $\cC_n(B)$ be a pseudo-magnifying, algebraic model over a connected 
graph, $B$.  
Then for any $i$ and $\theta>0$, there is an $r\in \naturals$ such that
for $\ell=\mu_1(B)$ and any $\nu\le \mu_1(B)$ we have
$$
\EE_{G\in\cC_n(B)}\Bigl[
\II_{{\rm TangleFree}(\ge\nu,<r)}(G)
\Bigl(\#
\bigl(\specnew_B(H_G)\cap B_{n^{-\theta}}(\ell) \bigr)\Bigr) \Bigr]
= O(n^{-i}) .
$$
\end{lemma}
This will give us a way to prove that $\tau_{\rm alg}=+\infty$ in our
basic models when
$B$ is Ramanujan.

\subsection{Results on Pseudo-Magnification}

Here is the main result we need.

\begin{lemma}\label{le_pseudo_magnification}
All of our basic models over a connected, pruned graph, $B$, with 
$\chi(B)>0$
are pseudo-magnifying.
\end{lemma}
We remark that the proof we give can be modified to work without the
condition that $B$ be pruned, but the assumption of being pruned
simplifies the proof (and in our applications, $B$ will be $d$-regular
for $d\ge 3$, so $B$ is necessarily pruned).

We prove this with a standard type of counting argument.
The case where $B$ has no half-loops is a bit simpler and illustrates
all the main ideas; hence
we first prove Lemma~\ref{le_pseudo_magnification} in this case.

\section{Pseudo-Magnification in Base Graphs Without Half-Loops}
\label{se_pseudomag_nohalf}

The point of this section is to prove that when $B$ has no half-loops,
then our (two) basic models over $B$ are pseudo-magnifying.

\begin{lemma}\label{le_pseudo_mag_no_half_loops}
Let $B$ be a connected, pruned graph without half-loops and with
$\chi(B)<0$.  Then the permutation and
full-cycle models over $B$ are pseudo-magnifying.
\end{lemma}

We will address the case where $B$ has half-loops in the next section;
however, the case where $B$ has no half-loops makes the
estimates simpler, and yet gives all the main ideas we will need
for the general case.  Hence we prove this special case first.

\subsection{The Counting Argument}

We will prove Lemma~\ref{le_pseudo_mag_no_half_loops} by a counting
argument.
Let us give basic definitions we need.

Our counting argument works as follows:  if $G\in\Coord_n(B)$,
then
$$
V_G = V_B \times [n].
$$
If such a $G$ is not an $(R,\gamma)$-pseudo-magnifier, then by
definition it follows that there
are sets
\begin{equation}\label{eq_sets_U_Uprime}
U \subset U' \subset V_G = V_B \times [n]
\end{equation} 
whose sizes satisfy
\begin{equation}\label{eq_sizes_U_Uprime}
R\le \#U \le \#V_G, \quad \#(U'\setminus U) = 
\lceil \gamma (\#U) \rceil -1 
\end{equation} 
where $\lceil \cdot \rceil$ denotes the ceiling function (the smallest
integer upper bound), such that
\begin{equation}\label{eq_gamma_U_in_Uprime}
\Gamma_G(U)\subset U' \ .
\end{equation} 
Our counting argument is the simple one: the probability that
$G\in\cC_n(B)$ is not a pseudo-magnifier is bounded by
$$
\sum_{U,U'} \Prob_{G\in\cC_n(B)}[\Gamma_G(U)\subset U']
$$
where we sum over each pair $U,U'$ satisfying
\eqref{eq_sets_U_Uprime} and \eqref{eq_sizes_U_Uprime};
we will show that for each $i$ there are $R,\nu$ such that
the above sum is bounded by $O(n^{-i})$.

Now we build up the tools we need.  We begin by setting
\begin{equation}\label{eq_p_U_U_prime}
p(U,U') \eqdef \Prob_{G\in\cC_n(B)}[\Gamma_G(U)\subset U'];
\end{equation} 
we now study $p(U,U')$.

\subsection{Almost Equal Fibre Sizes}

First we prove out that $p(U,U')=0$ unless
$U\subset V_B\times [n]$ has nearly equal
``fibre sizes.''  Let us make this precise.

\begin{definition}
Let $B$ be a connected graph, and 
let $U\subset V_B\times [n]$ for some $n\in\naturals$.
By the {\em $V_B$-fibres} (or simply {\em fibres}) {\em of $U$}
we mean the family of subsets of $[n]$,
$\{U_v\}_{v\in V_B}$, indexed on $v\in V_B$, defined by
$$
U_v \eqdef \{ i\in[n] \ | \ (v,i)\in U \} \subset [n].
$$
\end{definition}

\begin{lemma}\label{le_almost_equal_fibre}
Let $B$ be a connected graph.  Then for any $\epsilon\in(0,1)$ there is
a $\nu_1=\nu_1(\epsilon)>0$ for which the following is true:
for $n\in\naturals$ sufficiently large (depending only on $B,\epsilon$),
let $U\subset V_B\times [n]$ satisfy
\begin{equation}\label{eq_max_min_epsilon}
\min_{v\in V_B} \# U_B < (1-\epsilon) \max_{v\in V_B} \# U_B.
\end{equation} 
Then for any $G\in\Coord_n(B)$, 
\begin{equation}\label{eq_mag_fibres_diff}
\# \bigl( \Gamma_G(U) \setminus U \bigr) \ge \nu_1(\# U).
\end{equation} 
\end{lemma}
\begin{proof}
Let $m=\#V_B$, and let $\epsilon'>0$ be such that
$$
\bigl(1-\epsilon'\bigr)^{m-1} = 1-\epsilon.
$$
Let $v_{\rm min},v_{\rm max}$ be respective vertices where
$\#U_v$ (in the above definition)
takes its minimum and maximum values.  Since $B$ is connected
there is a path from $v_{\rm min}$ to $v_{\rm max}$ consisting of
vertices
$$
v_{\rm max}=v_1,v_2,\ldots,v_k=v_{\min} 
$$
with $k\le m-1$.  If \eqref{eq_max_min_epsilon} holds, then 
$$
\#U_{v_{\min}} < (1-\epsilon')^{m-1} \bigl(\#U_{v_{\max}}\bigr)
\le (1-\epsilon')^k \bigl(\#U_{v_{\max}}\bigr)
$$
and therefore
for some
$i\in [k-1]$ we must have
$$
\bigl(1-\epsilon'\bigr) \bigl( \# U_{v_{i+1}} \bigr)
< \# U_{v_i}  ;
$$
consider the smallest value of $i$ for which the above holds.  Then
we have
$$
\# U_{v_i} \ge \bigl(1-\epsilon'\bigr) \bigl( \# U_{v_{i-1}} \bigr)
\ge \cdots
\ge \bigl(1-\epsilon'\bigr)^{i-1} \bigl( \# U_{v_1} \bigr) 
\ge \bigl(1-\epsilon\bigr)\bigl( \# U_{v_1} \bigr) .
$$
%
Since there is an edge from $v_i$ to $v_{i+1}$, we have that
$\Gamma(U)$ has a fibre of size at least $\#U_{v_i}$ over $v_{i+1}$,
and hence
$$
\# \bigl( \Gamma_G(U) \setminus U \bigr) 
\ge
(\# U_{v_i}) - (\# U_{v_{i+1}})\ge 
\epsilon' (\# U_{v_i}) \ge \epsilon'(1-\epsilon)(\# U_{v_1}).
$$
Since $\# U \le m (\#U_{v_{\max}})=m(\#U_{v_1})$, applying this to the 
rightmost term above yields
$$
\# \bigl( \Gamma_G(U) \setminus U \bigr) 
\ge
\epsilon'(1-\epsilon)(1/m) (\# U).
$$
Hence the lemma holds, i.e., \eqref{eq_max_min_epsilon} holds, with
$$
\nu_1(\epsilon)= \epsilon' (1-\epsilon)/m > 0.
$$
\end{proof}

\subsection{The Probability Bound}

Next we give a simple bound for $p(U,U')$ in
\eqref{eq_p_U_U_prime}.

\begin{lemma}\label{le_permutation_prob}
Let $n\in\naturals$, and let $W,W'\subset [n]$ be subsets with
$\#W\le \# W'$.  If $\sigma\in\cS_n$ is a random permutation, then
$$
\Prob_{\sigma}[\sigma(W)\subset W']
= \frac{\binom{\#W'}{\#W}}{\binom{n}{\#W}} ,
$$
and if $\sigma$ is a random full-cycle then the above probability is
at most $n$ times the above right-hand-side.
\end{lemma}
\begin{proof}
The formula for $\sigma$ a random permutation is immediate.
Each random full-cycle occurs with probability $1/(n-1)!$, which
is exactly $n$ times its probability of occurring as a random
permutation; this implies the statement about the full-cycle case.
\end{proof}

\begin{corollary}\label{co_single_edge}
Let $B$ be a graph, and $\cC_n(B)$ be one of our basic models.
Let $n\in\naturals$, and let 
$U,U'\subset V_B\times [n]$.  
For each $v\in V_B$ let $s_v=\#U_v$ and $s'_v=\#U'_v$.
If $e\in E_B$ is not a half-loop,
and $\sigma$ is the permutation assignment $\Edir_B\to\cS_n$ associated
to a $G\in\cC_n(B)$, then the probability that 
$$
\sigma(e)U_{te}\subset U'_{he} \quad \mbox{and}\quad
\sigma(e^{-1})U_{he}\subset U'_{te}
$$
is $0$ if 
$s'_{he}<s_{te}$ or $s'_{te}<s_{he}$, and is otherwise
less than $p_1 p_2$, where 
$$
p_1 = \sqrt{  \frac{ n \binom{s'_{he}}{s_{te}}}{\binom{n}{s_{te}}} } ,
\quad
p_2 = \sqrt{  \frac{ n \binom{s'_{te}}{s_{he}}}{\binom{n}{s_{he}}} } .
$$
\end{corollary}
\begin{proof}
The statement about when the probability is zero is clear.
Otherwise,
if $p_1\le p_2$ then we have $\sigma(e)U_{te}\subset U'_{he}$ occurs
with probability at most $p_1^2\le p_1p_2$; similarly if
$p_1>p_2$ and for $\sigma(e^{-1})U_{he}\subset U'_{te}$.
\end{proof}

\subsection{A Binomial Coefficient Estimate}

In this section we introduce some useful formulas regarding binomial
coefficients, and prove a lemma that will be useful to us in the
counting argument we give.
This lemma, as is typical in our counting arguments, is straightforward
but involves some calculation.

First, 
Stirling's formula shows that for all integers $0\le b\le a$ we have
that
$$
C_1 2^{a H_2(b/a)} a^{-1/2}
\le
\binom{a}{b} \le C_2 2^{a H_2(b/a)}
$$
for some absolute constants $C_1,C_2$, where
$$
H_2(\mu) \eqdef -\mu\log_2\mu - (1-\mu)\log_2(1-\mu) .
$$
It follows that for $0\le b\le a$ and $a\ge 1$ we have
\begin{equation}\label{eq_H_two_approx}
\log_2\binom{a}{b} = a H_2(b/a) +O(\log_2 a)
\end{equation}
where the $O(\log_2 a)$ is bounded by an absolute constant 
(i.e., independent of $a,b$) times
$\log_2 a$.

We will also use the formula for the second derivative of $H_2(x)$
\begin{equation}\label{eq_H_second_deriv}
H_2''(x) = \frac{-\log_2 e}{x(1-x)}, \quad \forall x\in(0,1).
\end{equation}

\begin{lemma}\label{le_binom_coeff_estimate}
For any $C>0$ and $j\in\naturals$, 
for any sufficiently small $\theta>0$ the following holds:
there are natural numbers $S_0=S_0(\theta)$ and
$n_0=n_0(\theta)$ such that
for $n\ge n_0$ 
and any non-negative integers $s',s$
with $S_0\le s\le n(1/2 + \theta)$, and $s'\le \theta s$
we have
\begin{equation}\label{eq_binom_coeff_lemma}
\binom{n}{s'} \le n^{-j}\binom{n}{s}^{1/C}  .
\end{equation} 
\end{lemma}

\begin{proof}
Taking logs and dividing by $n$ it suffices to show that
\begin{equation}\label{eq_s_prime_j_prime}
H_2(s'/n) \le -j'\frac{\log_2 n}{n} + (1/C) H_2(s/n)
\end{equation} 
where $j'$ is $j$ plus constants to absorb the $O(\log_2 n)$ terms
in \eqref{eq_H_two_approx}
with $a=n$.
Hence it suffices to prove that all sufficiently small $\theta>0$,
there are $S_0,n_0$ such that
for all $x\in[S_0/n,1/2+\theta]$ we have
\begin{equation}\label{eq_g_versus_j_prime}
g(x) \ge j'\frac{\log_2 n}{n},
\quad\mbox{where}\quad
g(x) \eqdef  (1/C) H_2(x) - H_2(\theta x).
\end{equation} 
We shall do so by first showing that
for fixed $C>0$, for sufficiently small $\theta>0$ we have that
\begin{equation}\label{eq_g_concave_down}
g''(x) \le 0 \quad\forall x\in(0,1) .
\end{equation} 
It follows that to establish
\eqref{eq_g_versus_j_prime} it suffices to check this
at $x=S_0/n$ and $x=1/2+\theta$.

Using \eqref{eq_H_second_deriv} we have
$$
g''(x) \log_e 2
= \frac{-1}{Cx(1-x)}  + \frac{\theta^2}{\theta x (1-\theta x)}
=
\frac{-1}{Cx(1-x)}  + \frac{\theta}{x(1-\theta x)} .
$$
Note that 
$$
\frac{-1}{Cx(1-x)} \le \frac{-1}{Cx}, 
$$
and for $x\in(0,1)$ and $\theta\le 1/2$ we have $1-\theta x\ge 1/2$ and hence
$$
\frac{\theta}{x(1-\theta x)} \le \frac{2\theta}{x} .
$$
It follows that for $x\in(0,1)$ and $\theta\in(0,1/2)$ we have
$$
g''(x) \log_e 2 \le (2\theta-1/C) \frac{1}{x} \le 0
$$
provided that $2\theta \le 1/C$.
This establishes \eqref{eq_g_concave_down} for $\theta>0$ with
$\theta\le 1/(2C)$ and $\theta\le 1/2$.

So consider only those $\theta>0$ with
$$
\theta \le \min\bigl(1/(2C),1/4 \bigr).
$$
Then 
$g''(x)\le 0$ for all $x\in(0,1)$; it follows that
to prove \eqref{eq_g_versus_j_prime} for all $x\in [S_0/n,1/2+\theta]$,
it remains to show that 
$$
g(S_0/n),g(1/2+\theta) \ge j'\frac{\log_2 n}{n}
$$
for some some fixed $S_0$ and $n$ sufficiently large.

Since $\theta\le 1/4$, and since $H_2$ is monotone increasing on
$(0,1/2)$ and monotone decreasing on $(1/2,1)$, we have
$$
g(1/2+\theta) = (1/C) H_2(1/2+\theta) - H_2(\theta (1/2+\theta))
$$
\begin{equation}\label{eq_three_quarters}
\ge (1/C)H_2(3/4) - H_2(\theta(3/4))
\end{equation}
which is strictly positive for sufficiently small $\theta>0$
(since $H_2(x)\to 0$ as $x\to 0$, and $H_2(3/4)>0$).
For such a $\theta$ we have
$$
g(1/2+\theta) > 0
$$
and is therefore greater than $j'\log_2 n/n$ for $n$ sufficiently large.

So fix any $\theta>0$ with $\theta\le 1/2$, $\theta\le 1/(2C)$,
and for which \eqref{eq_three_quarters} is positive.
For any fixed constant $K$ we have 
$$
H_2(K/n)=( K\log_2 n + O(1) )/n
$$
for large $n$, and hence 
for fixed $S_0$ we have
$$
g(S_0/n) = (1/C) H_2(S_0/n) - H_2(\theta S_0/n)
= \bigl( (1/C) S_0 - \theta S_0 \bigr) \bigl( \log_2 n + O(1) \bigr) /n 
$$
for $n$ large.  Hence for any $S_0$ with
$$
S_0 \bigl( (1/C) - \theta \bigr) > j'
$$
we have 
$$
g(S_0/n) \ge  j'\frac{\log_2 n}{n}
$$
for sufficiently large $n$.  Since $\theta<1/(2C)$, we have
$1/C - \theta$ is positive; and hence the above inequality holds
for sufficiently large $n$ provided that
$$
S_0 > j'/(1/C-\theta).
$$

It follows that for $\theta$ and $S_0$ as above, we have
\eqref{eq_g_versus_j_prime} when $x$ is either endpoint of
$[S_0/n,1/2+\theta]$, and hence it holds for the entire interval.
\end{proof}

\subsection{Some Notation and Our Counting Lemma}

In this subsection we will introduce some helpful notation and
give a lemma that summarizes the counting argument
we shall use; the lemma is based on a simple union bound.

If $U\subset V_B\times [n]$ for some $n$, we use
${\rm Sizes}(U)$ to denote the function $V_B\to\integers_{\ge 0}$ given by
${\rm Sizes}(U)(v)=\#U_v$.

For any $\mec s\from V_B\to\integers_{\ge 0}$ we use the following notation:
\begin{equation}\label{eq_s_min_max_avg}
s_{\min} = \min_{v\in V_B} s_v, \quad
s_{\max} = \max_{v\in V_B} s_v, \quad
\overline s = (\#V_B)^{-1} \sum_{v\in V_B} s_v;
\end{equation}
and similarly for any $\mec s'\from V_B\to\integers_{\ge 0}$
(i.e., for $s_{\min}',s_{\max}',\overline s'$).

Here is a simple consequence of the union bound and
Lemma~\ref{le_almost_equal_fibre}.
In this lemma we use $v\sim u$ to denote the fact that $v,u$ are
adjacent vertices in $V_B$.

\begin{lemma}\label{le_counting_argument}
Let $B$ be a graph
(with or without half-loops), and $\cC_n(B)$ any model over $B$ (algebraic or 
not).  Set $m=\#V_B$.
Say that for any $i\in\naturals$ there are $R,\nu,\epsilon>0$ such that
the following holds:
for any $\mec s,\mec s'$ from $V_B\to\{0,1,\ldots,n\}$ such that
\begin{equation}\label{eq_conditions_s_s_prime}
s_v \le s_u' \quad\mbox{whenever $v=u$ or $v\sim u$} ,
\end{equation}
\begin{equation}\label{eq_conditions_sprime}
\mec s' \cdot \mec 1 \le (1+\nu)\mec s\cdot 1,
\quad s'_{\max}-s_{\min} \le n/2,
\end{equation}
\begin{equation}\label{eq_conditions_s_alone}
R \le \mec s\cdot \mec 1 \le nm/2,
\quad
s_{\min} \ge (1-\epsilon) s_{\max} ,
\end{equation} 
we have
\begin{equation}\label{eq_desired_counting_bound}
\max_{U,U'}\Bigl( \Prob_{G\in\cC_n(B)}\bigl[
\Gamma_G(U)\subset U' \bigr]  \Bigr)
\ \prod_{v\in V_B} 
\Biggl( \binom{n}{s_v}\binom{n}{s_v'-s_v} \Biggr)
= O(n^{-i})
\end{equation} 
where the above max is over all $U,U'$ such that
\begin{equation}\label{eq_conditions_sizes}
\quad {\rm Sizes}(U)=\mec s,
\quad {\rm Sizes}(U')=\mec s',
\quad U\subset U'.
\end{equation} 
Then $\cC_n(B)$ is pseudo-magnifying.
Similarly provided that \eqref{eq_desired_counting_bound} is
replaced with the bound
\begin{equation}\label{eq_simpler_desired_counting_bound}
\binom{n}{s_{\max}'-s_{\rm min}}^{\#V_B}
\max_{U,U'}\Bigl( \Prob_{G\in\cC_n(B)}\bigl[
\Gamma_G(U)\subset U' \bigr]  \Bigr)
\ \prod_{v\in V_B} 
\binom{n}{s_v}
= O(n^{-i}) .
\end{equation} 
\end{lemma}
\begin{proof}
Given an integer $i'\in\naturals$, let us find $R,\gamma>0$ such that
\begin{equation}\label{eq_prob_not_pseudomag}
\Prob_{G\in\cC_n(B)}
\bigl[ \mbox{$G$ is not a $(R,\gamma)$-pseudomagnifier}\bigr]
= O(n^{-i'}).
\end{equation} 

First, 
let $R,\nu,\epsilon>0$ be such that
\eqref{eq_desired_counting_bound} holds for $i=i'+2m$
whenever $\mec s,\mec s'$
satisfy 
\eqref{eq_conditions_s_s_prime},
\eqref{eq_conditions_sprime},
\eqref{eq_conditions_s_alone}.
Second, let $\nu',\epsilon'>0$ satisfy
\begin{equation}\label{eq_gamma_esp_prime_mystery}
\nu' m+ 1/(1-\epsilon') - (1-\epsilon') \le 1.
\end{equation} 
Let us show that \eqref{eq_prob_not_pseudomag} holds
with
\begin{equation}\label{eq_gamma_big_min}
\gamma = \min\bigl(\nu,\nu',\nu_1(\epsilon),\nu_1(\epsilon') \bigr).
\end{equation} 

The union bound implies that the probability that $G\in\cC_n(B)$
is not an $(R,\gamma)$-pseudomagnifier is at most
\begin{equation}\label{eq_union_bound}
\sum_{R\le \#U\le nm/2} p_1(n,U,\gamma) ,
\end{equation} 
where
$$
p_1(n,U,\gamma) \eqdef
\Prob_{G\in\cC_n(B)}
\bigr[ \#\bigl( \Gamma_G(U)\setminus U \bigr) \le \gamma (\# U )] .
$$
In view of \eqref{eq_gamma_big_min},
$$
p_1(n,U,\gamma) \le p_1\bigl(n,U,\nu_1(\epsilon) \bigr),
$$
and
Lemma~\ref{le_almost_equal_fibre} implies that
$$
p_1\bigl(n,U,\nu_1(\epsilon)\bigr) = 0
$$
whenever ${\rm Sizes}(U)=\mec s$ and
$$
s_{\min} < (1-\epsilon) s_{\max}.
$$
Hence in the union bound \eqref{eq_union_bound} we may restrict
the sum to those $\mec s$ with
$$
s_{\min} \ge  (1-\epsilon) s_{\max}.
$$
Since the number of possible $\mec s$ is (crudely) bounded by
$(n+1)^m$, to establish
\eqref{eq_union_bound}, it suffices to show that
for all $\mec s$ satisfying \eqref{eq_conditions_s_alone} we have
\begin{equation}\label{eq_union_bound_fixed_s}
\sum_{U,\ {\rm Sizes}(U)=\mec s} p_1(n,U,\gamma) \le O(n^{-i-m}) .
\end{equation} 

Next note that for any $U$ and $G\in\cC_n(B)$ for which
$$
 \#\bigl( \Gamma_G(U)\setminus U \bigr) \le \gamma (\# U ) ,
$$
the set $U'=\Gamma_G(U)\cup U$ satisfies
$$
\Gamma_G(U) \subset U' ;
$$
moreover setting $\mec s'={\rm Sizes}(U')$ then $\mec s'$ must satisfy
$$
\mec s\le \mec s', \quad 
\mec s' \cdot \mec 1 \le (1+\gamma)\mec s\cdot \mec 1
\le (1+\nu) \mec s\cdot\mec 1 ,
$$
and if $u\sim v$ then $\Gamma_G(U)\subset U'$ implies that
$s_v \le s_u'$ (or else the $v$-fibre over $U$ could not ``fit into''
the $u$-fibre over $U'$ under $G$-adjacency).

Once we fix $\mec s,\mec s'$, the number of $U\subset U'$ with those
respective sizes is exactly
$$
\prod_{v\in V_B} \binom{n}{s_v}\binom{n-s_v}{s_v'-s_v}, 
$$
which is bounded from above by 
$$
\prod_{v\in V_B} \binom{n}{s_v}\binom{n}{s_v'-s_v} .
$$
Since there are (crudely, again) at most $(n+1)^m$ possible values for
$\mec s'$, we therefore have that 
to prove \eqref{eq_union_bound_fixed_s}, it suffices to show that
for all $\mec s,\mec s'$ as above we have
\begin{equation}\label{eq_union_two_m}
\max_{U,U'}\Bigl( \Prob_{G\in\cC_n(B)}\bigl[
\Gamma_G(U)\subset U' \bigr]  \Bigr) 
\ \prod_{v\in V_B} \binom{n}{s_v}\binom{n}{s_v'-s_v} 
= O(n^{-i'-2m}) = O(n^{-i}).
\end{equation} 
where the $\max$ is over all $U,U'$ such that
\eqref{eq_conditions_sizes} holds.

In view of \eqref{eq_gamma_big_min}, we see that in
\eqref{eq_union_two_m} it suffices to sum over
$\mec s',\mec s$ that additionally satisfy
$$
\mec s'\cdot \mec 1 \le (1+\nu')\mec s\cdot \mec 1,
\quad s_{\min} \ge (1-\epsilon') s_{\max}.
$$
However, for such $\mec s',\mec s$ we claim that
$$
s'_{\max}-s_{\min} \le n/2 ;
$$
indeed, 
$$
s'_{\max} - s_{\min} 
\le
(s'_{\max}-s_{\max}) + (s_{\max}-s_{\min})
\le
\gamma' m\overline s + \bigl( \overline s/(1-\epsilon') - 
\overline s(1-\epsilon') \bigr) 
$$
and using $\overline s\le n/2$ we conclude that
$$
s'_{\max} - s_{\min}
\le
n/2 \bigl(\gamma' m+ 1/(1-\epsilon') - (1-\epsilon') \bigr)
\le n/2.
$$
Hence we may also limit \eqref{eq_union_two_m} to those $\mec s,\mec s'$
for which
$$
s'_{\max} - s_{\min}\le n/2,
$$
and hence in \eqref{eq_desired_counting_bound} we may restrict
our consideration to $\mec s,\mec s'$ satisfying
\eqref{eq_conditions_s_s_prime}--\eqref{eq_conditions_s_alone}.

For the statement regarding 
\eqref{eq_simpler_desired_counting_bound}, notice that since our
restrictions on $\mec s,\mec s'$ include 
$s_{\max}'-s_{\min} \le n/2$, for any $v$ we have
$$
\binom{n}{s_v'-s_v} \le \binom{n}{s_{\max}'-s_{\min}},
$$
and hence for all relevant $\mec s,\mec s'$ the bound
\eqref{eq_simpler_desired_counting_bound} implies
\eqref{eq_desired_counting_bound}.
\end{proof}

\subsection{Proof of Lemma~\ref{le_pseudo_mag_no_half_loops}}
\label{su_when_no_half_loops}

\begin{proof}[Proof of Lemma~\ref{le_pseudo_mag_no_half_loops}]
According to Lemma~\ref{le_counting_argument}, it suffices to
show that for each $i\in\naturals$ there exist $R,\nu,\epsilon>0$
such that 
\eqref{eq_desired_counting_bound} holds for all
$\mec s,\mec s'$ satisfying
\eqref{eq_conditions_s_s_prime}--\eqref{eq_conditions_s_alone}.
So fix an $i\in\naturals$, and let us seek such $R,\gamma,\epsilon$.

According to Corollary~\ref{co_single_edge}, 
\begin{equation}\label{eq_probability_Gamma_U_in_U_prime_no_half}
\Prob_{G\in\cC_n(B)}\Bigl[ \Gamma_G(U)\subset U' ] \le
\prod_{v\in V_B} \prod_{u\sim v}
\sqrt{  \frac{ n \binom{s'_{u}}{s_{v}}}{\binom{n}{s_{v}}} } 
\end{equation} 
where $u\sim v$ is shorthand for multiplication over all $e$ with
$te=v$ of $u=he$ (hence for multiple edges the factor
of $u$ in the product occurs multiple times).  
Since
$$
\prod_{u\sim v}
\sqrt{  \frac{ n \binom{s'_{u}}{s_{v}}}{\binom{n}{s_{v}}} }
\le
\left(  \frac{ n \binom{s'_{\max}}{s_v}}{\binom{n}{s_{v}} }
\right)^{\deg_B(v)/2} 
\le
\left(  \frac{ n \binom{s'_{\max}}{s'_{\max}-s_v}}{\binom{n}{s_{v}} }
\right)^{\deg_B(v)/2} 
\le
\left(  \frac{ n \binom{n}{s'_{\max}-s_v}}{\binom{n}{s_{v}} }
\right)^{\deg_B(v)/2} ,
$$
we have
\begin{equation}\label{eq_prob_estimate_good_enough}
\Prob_{G\in\cC_n(B)}\Bigl[ \Gamma_G(U)\subset U' ] \le
\prod_{v\in V_B} 
\left( \frac{ n \binom{n}{s'_{\max}-s_v}}{\binom{n}{s_v}} 
\right)^{\deg_B(v)/2}.
\end{equation} 
Hence the left-hand-side of
\eqref{eq_simpler_desired_counting_bound}, namely
$$
\binom{n}{s_{\max}'-s_{\rm min}}^{\#V_B}
\max_{U,U'}\Bigl( \Prob_{G\in\cC_n(B)}\bigl[
\Gamma_G(U)\subset U' \bigr]  \Bigr)
\ \prod_{v\in V_B} 
\binom{n}{s_v}
$$
is bounded above by
\begin{equation}\label{eq_expression_to_bound}
n^{C_1} \binom{n}{s_{\max}'-s_{\rm min}}^{C_2}
\prod_{v\in V_B} \binom{n}{s_v}^{\bigl( 2-\deg_B(v) \bigr)/2}
\end{equation} 
for constants $C_1,C_2>0$.
Since $B$ is pruned, $2-\deg_B(v)\le 0$ for all $v\in V_B$, and since
$$
\chi(B) = \sum_{v\in B}  \bigl( 2-\deg_B(v) \bigr)/2
$$
is negative, 
$2-\deg_B(v)\le -1$ for at least one $v$.  Hence
to establish \eqref{eq_simpler_desired_counting_bound}, it suffices
to show that for any $i$ there are $R,\nu,\epsilon>0$ such that
for any $v\in V_B$ we have
\begin{equation}\label{eq_simpler_expression_to_bound}
n^{C_1} \binom{n}{s_{\max}'-s_{\rm min}}^{C_2}
\binom{n}{s_v}^{-1/2}
=O(n^{-i})
\end{equation} 
provided that $\mec s',\mec s$ satisfy
\eqref{eq_conditions_s_s_prime}--\eqref{eq_conditions_s_alone}.

(Since $B$ has no half-loops, $\chi(B)$ is an integer, so we
could replace the $-1/2$ exponent in \eqref{eq_simpler_expression_to_bound}
by $-1$; in the next section, when $B$ may have half-loops, we cannot
do so, and the $-1/2$ exponent will be sufficient for us.)

So let us apply Lemma~\ref{le_binom_coeff_estimate} with $C=2C_2$
and any $j\ge (i+C_1)/C_2$; fix a $\theta>0$ sufficiently small so
that there exist $S_0,n_0$ for which
\begin{equation}\label{eq_binom_coeff_C_i}
\binom{n}{r} \le n^{-(i+C_1)/C_2}  \binom{n}{s}^{1/(2 C_2)}
\end{equation} 
provided that
\begin{equation}\label{eq_binom_coeff_conditions}
n\ge n_0, \quad
S_0\le s\le n(1/2+\theta), \quad 
r\le \theta s;
\end{equation} 
for such $n,r,s$ we therefore have
$$
\binom{n}{s}^{-1/2} \binom{n}{r}^{C_2} n^{C_1} \le n^{-i}.
$$
It follows that \eqref{eq_simpler_expression_to_bound} holds 
provided that for all $v$ we have
$$
S_0 \le s_v \le n(1/2+\theta) , \quad
s_{\max}' - s_{\min} \le \theta s_v
$$
or, in other words,
\begin{equation}\label{eq_needed_for_s_sprime}
S_0 \le s_{\min},\quad
s_{\max} \le n(1/2+\theta) , \quad
s_{\max}' - s_{\min} \le \theta s_{\min}.
\end{equation} 

Now we specify $R,\nu,\epsilon$: take $\nu=\theta/(2m)$, choose
any $\epsilon>0$ sufficiently small so that
$$
(\theta/2 + \epsilon)/(1-\epsilon) \le \theta
$$
and
$$
\frac{1/2}{1-\epsilon}\le 1/2+\theta,
$$
and finally any $R\in \naturals$ with
$$
R \ge S_0 m / (1-\epsilon).
$$
If $\mec s,\mec s'$ satisfy
\eqref{eq_conditions_s_s_prime}--\eqref{eq_conditions_s_alone}
with these values of $R,\nu,\epsilon$, let us verify the three 
conditions in
\eqref{eq_needed_for_s_sprime} hold:
\begin{enumerate}
\item $S_0\le s_{\min}$:
$$
R \le \mec s\cdot\mec 1\le m s_{\max}\le m s_{\min}/(1-\epsilon),
$$
so $s_{\min}\ge R(1-\epsilon)/m\ge S_0$.
\item $s_{\max}\le n(1/2+\theta)$:
$$
s_{\max}\le s_{\min}/(1-\epsilon) \le \overline s/(1-\epsilon)
\le (n/2)/(1-\epsilon)\le n(1/2+\theta).
$$
\item $s_{\max}'-s_{\min} \le \theta s_{\min}$:
for all $v\in V_B$ we have
$$
s_v' - s_v \le 
\mec s' \cdot \mec 1
-
\mec s \cdot \mec 1 \le \nu\, \mec s\cdot \mec 1 = \nu m \overline s
(\theta/2)
\le (\theta/2) s_{\min} / (1-\epsilon)
$$
and hence (taking $v$ maximizing $s_v'$)
$$
s_{\max}' \le s_{\max} + (\theta/2) s_{\min}/(1-\epsilon)
\le s_{\min}/(1-\epsilon) +  (\theta/2)s_{\min}/(1-\epsilon) ,
$$
so
$$
s_{\max}' - s_{\min} \le 
s_{\min} (\theta/2 + \epsilon)/(1-\epsilon) \le \theta s_{\min}
$$
\end{enumerate}

Hence for these values of $R,\nu,\epsilon$ we have
\eqref{eq_needed_for_s_sprime}, and therefore
hence \eqref{eq_simpler_expression_to_bound},
and therefore \eqref{eq_simpler_desired_counting_bound} in
Lemma~\ref{le_counting_argument}.
This lemma then implies that $\cC_n(B)$ is
pseudo-magnifying.
\end{proof}

\section{Pseudo-Magnification in Graphs With Half-Loops}
\label{se_pseudomag_with_half}

Let us now describe the ingredients needed to prove pseudo-magnification
when our base graphs may have half-loops.
The main point is that we have to include probability estimates
involving random
involutions.

\subsection{An Involution Probability Bound}

Here is the alternate of Lemma~\ref{le_permutation_prob} that we will
use; if $n,t\in\naturals$ with $t$ even and $t\le n$, we use the
``odd binomial coefficient notation:''
$$
\binom{n}{t}_{\rm odd} \eqdef 
\frac{(n-1)(n-3)\ldots(n-2t+1)}{(2t-1)(2t-3)\ldots 1}.
$$

\begin{lemma}\label{le_involution_prob}
Let $n\in\naturals$, and let $W\subset W'\subset [n]$ be subsets with
$\#W\le \# W'$.  
Let $s''$ be the largest non-negative even integer with
$s''\le 2(\#W)-(\#W')-1$.
If $\sigma\in\cS_n$ is a random perfect matching for
$n$ even, and a random near perfect matching for $n$ odd, then
$$
\Prob_{\sigma}[\sigma(W)\subset W']
\le 
\frac{\binom{\#W}{s''}}{\binom{n}{s''}_{\rm odd}} .
$$
\end{lemma}
\begin{proof}
Consider a $\sigma$ for which $\sigma(W)\subset W'$.
Note that 
$\sigma$ matches every element in $W$ with some element of $W'$, except
for possibly one element of $W$ (when $n$ is odd);
since at most $(\#W')-(\#W)$ elements of $W$ can be matched with 
elements in $W'\setminus W$, and at most one element of $W$ can
be matched with itself, it follows that the subset, $W''=W''(\sigma)\subset W$,
of elements that $\sigma$ matches in $W$ is of size
at least $s''$.
Since there are
$\binom{\#W}{s''}$ possible values of $W''=W''(\sigma)$, the union
bound implies that
$$
\Prob_{\sigma}[\sigma(W)\subset W']
\le 
\binom{\#W}{s''} p(n,s''),
$$
where $p(n,s'')$ is the probability that a random involution
$\sigma\in\cS_n$ matches a fixed subset, $W''$, of size $s''$ in pairs;
this probability equals the probability that a fixed element of $W''$
is matched with another element of $W''$, times the probability
that a fixed remaining element is paired with another remaining element,
etc.
Hence
$$
p(n,s'') = \frac{s''-1}{n-1}\frac{s''-3}{n-3}\ldots
\frac{3}{n-s''+3}\frac{1}{n-s''+1}
=1 \biggm/ \binom{n}{s''}_{\rm odd}
$$
\end{proof}

\subsection{Odd Binomial Coefficient Estimates}

It is be simpler for us to express odd binomial coefficients in terms of
almost equal expressions involving binomial coefficients.

\begin{lemma}\label{le_odd_binom}
Let $n,t\in\naturals$ with $t$ even.  Then
$$
\frac{n-t}{n}\binom{n}{t} \le
\left( {\binom{n}{t}}_{\rm odd}  \right)^2
\le
t\binom{n}{t} .
$$
\end{lemma}
\begin{proof}
Comparing factor by factor, we have
$$
(t-1)! \le 
\bigl( (t-1)(t-3)\ldots 1  \bigr)^2 \le t!
$$
and 
$$
(n-1)(n-2)\ldots (n-t)
\le
\bigl( (n-1)(n-3)\ldots(n-2t+1) \bigr)^2
\le
n(n-1)\ldots (n-t+1) \ .
$$
Dividing the second two inequalities by the first two yields
$$
\frac{(n-1)(n-2)\ldots (n-t)}{t!} 
\le \left( {\binom{n}{t}}_{\rm odd}  \right)^2
\le
\frac{n(n-1)\ldots (n-t+1)}{(t-1)!}
$$
which is equivalent to the upper and lower bounds in the lemma.
\end{proof}

\subsection{Additional Binomial Coefficient Estimates}

We will also use some easy binomial coefficient estimates.
First, for any $0\le r'\le r\le n$ we have
$$
\binom{n}{r} = \binom{n}{r'}\binom{n-r'}{r-r'}
\le \binom{n}{r'}\binom{n}{r-r'} ,
$$
and hence
\begin{equation}\label{eq_easy_binom_est}
\binom{n}{r'}^{-1/2} \le \binom{n}{r}^{-1/2} \binom{n}{r-r'}^{1/2}.
\end{equation} 
We will need the trivial estimate that for $r\ge 0$ and any $n$
(including $n\le r+1$)
\begin{equation}\label{eq_trivial_binom_est}
\binom{n}{r+1} \le \binom{n}{r} n,\quad
\binom{n}{r+2} \le \binom{n}{r} n^2.
\end{equation}

\subsection{Proof of Lemma~\ref{le_pseudo_magnification}}

\begin{proof}[Proof of Lemma~\ref{le_pseudo_magnification}]
According to Lemma~\ref{le_counting_argument}, it suffices to
show that for each $i\in\naturals$ there exist $R,\nu,\epsilon>0$
such that
\eqref{eq_desired_counting_bound} holds for all
$\mec s,\mec s'$ satisfying
\eqref{eq_conditions_s_s_prime}--\eqref{eq_conditions_s_alone}.
So fix an $i\in\naturals$, and let us prove that such $R,\gamma,\epsilon$
exist;
we shall reduce this proof to part of the proof given in
Subsection~\ref{su_when_no_half_loops} (the case where $B$ has no
half-loops).

For each $v\in V_B$, let ${\rm half}(v)$ be the number of half-loops
in $V$ about $v$.
Let $n\in\naturals$, and let $U\subset U'\subset V_B\times [n]$;
let $\mec s,\mec s'$ denote the fibre sizes of $U,U'$ respectively,
and assume that they satisfy
\eqref{eq_conditions_s_s_prime}--\eqref{eq_conditions_s_alone}
for some $R,\nu,\epsilon>0$ that we will later specify.
Let us add the assumptions that
\begin{equation}\label{eq_theta_like_assumption}
s'_{\max}-s_{\min} \le s_{\min}/3.
\end{equation}

From these assumptions on $\mec s,\mec s'$, we have
that for all $v\in V_B$ 
\begin{equation}\label{eq_convenient_half_assumption}
2s_v-s'_v  \ge 2s_{\min}-s'_{\max}\ge (2-4/3)s_{\min}\ge 0 .
\end{equation} 
For each $v\in V_B$, set $s_v''$ to be the largest even integer less than
$2s_v-s'_v$; hence
\begin{equation}\label{eq_s_v_double_prime_inequ}
2s_v-s'_v-2 \le s''_v \le 2s_v-s'_v -1.
\end{equation} 
Next, let us verify that for sufficiently large $n$ we have
\begin{equation}\label{eq_slight_further_assumption}
s'_{\max} - s_{\min} \le (n/2)-2, \quad
s_{\max} \le 2n/3;
\end{equation} 
the first inequality follows from
$$
s'_{\max}-s_{\min} \le s_{\min}/3 \le \overline s/3 \le n/6
$$
which is at most $n/2-2$ for $n\ge 6$; the second inequality follows from
$$
s_{\max}\le s'_{\max} \le (4/3)s_{\min} \le (4/3)(n/2) = 2n/3.
$$

According to Lemmas~\ref{le_involution_prob} and
\ref{le_permutation_prob},
\begin{equation}\label{eq_prob_bounded_by_Q_prod}
\Prob_{G\in\cC_n(B)}\bigl[
\Gamma_G(U)\subset U' \bigr] \le 
\ \prod_{v\in V_B}
Q(v)
\end{equation} 
where
$$
Q(v) = 
\left( 
\frac{\binom{s_v}{s_v''}}{\binom{n}{s_v''}_{\rm odd}} 
\right)^{{\rm half}(v)}
\prod_{u\sim v}
\sqrt{  \frac{ n \binom{s'_{u}}{s_{v}}}{\binom{n}{s_{v}}} }
$$
where $u\sim v$ is the product over all edges $e$ that are 
not half-loops, and whose tail is $u$ and whose head is $v$.
Let us show that for each $v\in V_B$ we have
\begin{equation}\label{eq_new_Q_estimate}
Q(v) \le n^{K_1} \binom{n}{s'_{\max}-s_{\min}}^{K_2}
\binom{n}{s_v}^{\bigl( \deg_B(v)-2 \bigr)}
\end{equation} 
where $K_1=K_1(v),K_2=K_2(v)$ are constants depending on $v$;
once we do this we will obtain the same estimate as in
\eqref{eq_expression_to_bound}---with different 
constants $C_1,C_2$---and then finish the proof as it is finished
below \eqref{eq_expression_to_bound} (for the case there, where $B$
has no half-loops).

Since the number of such $e$ is $\deg_B(v)-{\rm half}(v)$,
\eqref{eq_prob_estimate_good_enough} and the equation
above it imply that
$$
Q(v) \le
\left( 
\frac{\binom{s_v}{s_v''}}{\binom{n}{s_v''}_{\rm odd}} 
\right)^{{\rm half}(v)} 
\left( \frac{ n \binom{n}{s'_{\max}-s_v}}{\binom{n}{s_v}}
\right)^{\bigl( \deg_B(v)-{\rm half}(v)\bigr)/2}
$$
To estimate the new term raised to the power ${\rm half}(v)$
(whenever ${\rm half}(v)>0$), we note that
\begin{equation}\label{eq_v_minus_v_prime}
s_v-s_v'' \le s'_v+2-s_v \le s'_{\max}-s_{\min}+2\le n/2
\end{equation} 
(using \eqref{eq_slight_further_assumption}), and therefore
$$
\binom{s_v}{s_v''} = \binom{s_v}{s_v-s_v''} \le
\binom{n}{s_v-s_v''} \le \binom{n}{s_v-s_v''}
\le \binom{n}{s'_{\max}-s_{\min}+1} 
\le \binom{n}{s'_{\max}-s_{\min}}n
$$
(using \eqref{eq_trivial_binom_est}).
Also, in view of Lemma~\ref{le_odd_binom} 
$$
\binom{n}{s_v''}_{\rm odd}^{1/2} \ge
\frac{n-s_v''}{n} \binom{n}{s_v''}.
$$
Hence
\begin{equation}\label{eq_new_prob_term}
\frac{\binom{s_v}{s_v''}}{\binom{n}{s_v''}_{\rm odd}} 
\le \binom{n}{s'_{\max}-s_{\min}} n
\left( \frac{n-s_v''}{n} \binom{n}{s_v''}
\right)^{-1/2} .
\end{equation} 
We also have
$$
s_v''\le 2s_v-s_v'\le s_v
$$
and so
\begin{equation}\label{eq_frac_n_s_v_primeprime}
\frac{n-s_v''}{n}\ge \frac{n-s_v}{n} .
\end{equation} 
Also, setting $r'=s_v''$ and $r=s_v$ in \eqref{eq_easy_binom_est} we have
$$
\binom{n}{s_v''}^{-1/2} \le \binom{n}{s_v}^{-1/2} \binom{n}{s_v-s_v''}^{1/2} ;
$$
using \eqref{eq_v_minus_v_prime} this implies
$$
\binom{n}{s_v''}^{-1/2} \le 
\binom{n}{s_v}^{-1/2} \binom{n}{s'_{\max}-s_{\min}+ 2}^{1/2}
\le
\binom{n}{s_v}^{-1/2} \binom{n}{s'_{\max}-s_{\min}}^{1/2} n.
$$
Applying this inequality and \eqref{eq_frac_n_s_v_primeprime}
to \eqref{eq_new_prob_term} we get
$$
\frac{\binom{s_v}{s_v''}}{\binom{n}{s_v''}_{\rm odd}} 
\le
n^2 \frac{n-s_v}{n} \binom{n}{s'_{\max}-s_{\min}}^{3/2}
\binom{n}{s_v}^{-1/2}.
$$
This establishes \eqref{eq_new_Q_estimate}.  It follows from
\eqref{eq_prob_bounded_by_Q_prod} that
$$
\Prob_{G\in\cC_n(B)}\bigl[
\Gamma_G(U)\subset U' \bigr] \le
n^{K_1'} \binom{n}{s'_{\max}-s_{\min}}^{K_2'}
\ \prod_{v\in V_B}
\binom{n}{s_v}^{\bigl(2-\deg_B(v)\bigr)/2},
$$
where $K_1',K_2'$ are the sum over the $K_1(v),K_2(v)$.
Hence for some $v\in V_B$ we have
$$
\Prob_{G\in\cC_n(B)}\bigl[
\Gamma_G(U)\subset U' \bigr] \le
n^{K_1'} \binom{n}{s'_{\max}-s_{\min}}^{K_2'}
\binom{n}{s_v}^{-1/2}
$$
(see the discussion between \eqref{eq_expression_to_bound}
and \eqref{eq_simpler_expression_to_bound}).
It follows that the left-hand-side of 
\eqref{eq_simpler_desired_counting_bound} is bounded by
\eqref{eq_expression_to_bound}
for some $C_1,C_2$ (involving $K_1',K_2'$).

Now we mimic the rest of the proof of Lemma~\ref{le_pseudo_mag_no_half_loops}
following
\eqref{eq_simpler_expression_to_bound}: to
establish \eqref{eq_simpler_desired_counting_bound} it suffices to
prove that for any $i$ there are $R,\nu,\epsilon>0$ that guarantee
\eqref{eq_simpler_expression_to_bound}.
We use Lemma~\ref{le_counting_argument} to find
$\theta>0$, $S_0,n_0$ such that
\eqref{eq_binom_coeff_C_i} holds provided that
\eqref{eq_binom_coeff_conditions} holds;
since Lemma~\ref{le_counting_argument} implies that there are
$S_0,n_0$ for arbitrarily small $\theta>0$, we may insist that
$\theta\le 1/3$.
Then we take $R,\nu,\epsilon>0$ as given just below
\eqref{eq_needed_for_s_sprime}, and the same proof shows that 
for any 
$\mec s,\mec s'$ that satisfy
\eqref{eq_conditions_s_s_prime}--\eqref{eq_conditions_s_alone},
then \eqref{eq_needed_for_s_sprime} holds;
since we took $\theta\le 1/3$, the last inequality in
\eqref{eq_needed_for_s_sprime} implies that
\eqref{eq_theta_like_assumption} holds.
It follows that for any $i$ there are $R,\nu,\epsilon>0$ such that
\eqref{eq_expression_to_bound} holds, and therefore
\eqref{eq_simpler_desired_counting_bound} holds.
Therefore Lemma~\ref{le_counting_argument} implies that
$\cC_n(B)$ is pseudo-magnifying.
\end{proof}

\section{Proof of Theorem~\ref{th_second_main_theorem}}
\label{se_proof_Ram}

In this section we prove Lemma~\ref{le_no_coeff_largest_base},
and then easily deduce Theorem~\ref{th_second_main_theorem}.

Our proof really shows that if $B$ is $d$-regular and
$\cC_n(B)$ is pseudo-magnifying, then $\tau_{\rm alg}$ can only be
finite if for some eigenvalue $\ell$ of the model, with $|\ell|>(d-1)^{1/2}$
but $\ell\ne d-1$,
we have that
\eqref{eq_new_eigenvalues_near_ell} holds for fixed $\tau$ and
arbitrarily small $\nu>(d-1)^{1/2}$ and arbitrarily large $r$.

\begin{proof}[Proof of Lemma~\ref{le_no_coeff_largest_base}]
For any $i$, the are $R,\gamma$ such that the $G\in\cC_n(B)$
probability that $G$ is not an $(R,\gamma)$-pseudo-magnifier is
at most $O(n^{-i})$.
However, for any $G\in\Coord_n(B)$ and $R$, if $U\subset V_G$
has $\#U\le R$, then either:
\begin{enumerate}
\item $\Gamma_G(U)\setminus U=\emptyset$, in which case $U$ is
disconnected from the other vertices in $V_G$, or
\item $\Gamma_G(U)\setminus U$ in nonempty, in which case
$$
\#\bigl( \Gamma_G(U)\setminus U \bigr) \ge 1 \ge (1/R)(\#U).
$$
\end{enumerate}
Hence setting
$$
\gamma'=\min(\gamma,1/R),
$$
if $G$ is not $\gamma'$-magnifier, then $G$ has a
connected component, $U'$, of size at most $R$, which makes
$U'$ a $\ge (d-1)$-tangle whose order is
$$
\ord(U') = \ord(B) \frac{\#V_{U'}}{\#V_B} \le \ord(B) R/ (\#V_B).
$$
Hence taking $r\in\naturals$ with $r\ge 1+\ord(B)R/(\#V_B)$,
the probability that $G$ is $(\ge d-1,<r)$-tangle free
and is not an $\gamma'$-magnifier is at most
$O(n^{-i})$.  

According to Alon's theorem, if $G\in\Coord_n(B)$ is a $\gamma'$-magnifier,
then all eigenvalues of $A_G$ except the largest one are bounded
away from $d$; it follows that $\specnew_B(H_G)$
is bounded away from $d-1$, and hence
$$
\specnew_B(H_G)\cap B_{n^{-\theta}}(d-1) 
$$
is empty for $n$ sufficiently large.  Since
the number of new eigenvalues of $H_G$ is $O(n)$,
it follows that for the above value of $r$ we have
$$
\EE_{G\in\cC_n(B)}\Bigl[
\II_{{\rm TangleFree}(\ge d-1,<r)}(G)
\Bigl(\#
\bigl(\specnew_B(H_G)\cap B_{n^{-\theta}}(\ell) \bigr)\Bigr) \Bigr]
= O(n^{-i+1}) .
$$
Since $i$ is arbitrary, we conclude Lemma~\ref{le_no_coeff_largest_base}.
\end{proof}

\begin{proof}[Proof of Theorem~\ref{th_second_main_theorem}]
If $\tau_{\rm alg}$ is finite, then for some $\tau\in\naturals$
and $\nu>(d-1)^{1/2}$ sufficiently small and $r\in\naturals$ sufficiently
large we have
\begin{equation}\label{eq_expected_near_d}
\EE_{G\in\cC_n(B)}\Bigl[
\II_{{\rm TangleFree}(\ge\nu,<r)}(G)
\Bigl(\#
\bigl(\specnew_B(H_G)\cap B_{n^{-\theta}}(\ell) \bigr)\Bigr) \Bigr]
= C_\ell n^{-\tau} + o(n^{-\tau}) .
\end{equation} 
for some $C_\ell>0$ ($C_\ell$ may depend on $\nu$ and $r$),
where $\ell$ an eigenvalue of the model with $|\ell|>(d-1)^{1/2}$.

First consider the case where $B$ is $d$-regular Ramanujan and
not bipartite.  Then $\ell=d-1$ is the only eigenvalue of the model with
$|\ell|>(d-1)^{1/2}$.  In view of Lemma~\ref{le_no_coeff_largest_base},
for any $i\in\naturals$,
the left-hand-side of \eqref{eq_expected_near_d} is $O(n^{-i})$
for $\nu=d-1$ and some $r\in\naturals$, which contradicts
\eqref{eq_expected_near_d} for these values of $\nu,r$, and
hence for any smaller $\nu$ and any larger $r$.
Hence $\tau_{\rm alg}=+\infty$.

Next consider the case where $B$ is connected, $d$-regular, Ramanujan,
and bipartite.  Then for any $G\in\Coord_n(B)$, $G$ is also bipartite,
and hence has the same multiplicity of the eigenvalue $d-1$ in $H_G$
as it does that of $-(d-1)$.  Since $H_G$ has one old eigenvalue
of $d-1$ and one of $-(d-1)$ (since $B$ is bipartite and connected),
it follows that
the left-hand-side of \eqref{eq_expected_near_d} is the
same for $\ell=d-1$ and $\ell=-(d-1)$.  
Since \eqref{eq_expected_near_d} cannot hold for $\ell=d-1$ with
$\nu$ arbitrarily close to $(d-1)^{1/2}$ and $r$ arbitrarily large---by
the argument in the previous paragraph---it also cannot hold for
$\ell=-(d-1)$.
Hence $\tau_{\rm alg}=+\infty$ also in this case.
\end{proof}

\section{Bounds on $\tau_{\rm tang}$}
\label{se_tau_tangle}

If $B$ is $d$-regular and Ramanujan, we now know that
Theorem~\ref{th_rel_Alon_regular} holds with
$$
\tau_1=\tau_2=\tau_{\rm tang},
$$
and therefore for fixed $\epsilon>0$ sufficiently small
we have matching upper and lower bounds proportional to 
$n^{-\tau_{\rm tang}}$ for
$$
\Prob_{G\in\cC_n(B)}\bigl[
{\rm NonAlon}_d(G;\epsilon)>0
\bigr] .
$$
It therefore becomes interesting to have bounds on $\tau_{\rm tang}$.
Let us just give those bounds that follow
from \cite{friedman_alon}.

\subsection{Lower Bounds on $\tau_{\rm tang}$ for any Model}

Chapter~6 of \cite{friedman_alon} computes $\tau_{\rm fund}$ of various
models, which is the smallest order of a $\ge (d-1)^{1/2}$-tangle.
By definition, $\tau_{\rm tang}$ is the smallest order of a 
$>(d-1)^{1/2}$-tangle, so the same techniques apply.  

\begin{theorem}\label{th_tau_tang_whole}
Let $B$ be a graph and let $m=m(B)$ be the
smallest integer with
$$
2m-1 > \mu_1^{1/2}(B).
$$
Then for any algebraic model, $\cC_n(B)$,
$\tau_{\rm tang}\ge m-1$, and equality holds if the bouquet
of $m$ whole-loops occurs in $\cC_n(B)$.
In particular, if $B$ is $d$-regular then 
($\mu_1(B)=d-1$ and) $m(B)$ depends only on $d$ and is given
by
\begin{equation}\label{eq_m_d_lower_bound_d}
m=m(d) = \Bigl\lfloor \bigl( (d-1)^{1/2}+1 \bigr)/2 \Bigr\rfloor + 1,
\end{equation} 
and
\begin{equation}\label{eq_tau_tang_lower_bound_d}
\tau_{\rm tang} \ge \Bigl\lfloor\bigl( (d-1)^{1/2}+1 \bigr)/2 \Bigr\rfloor;
\end{equation} 
furthermore, equality holds in the permutation model if
$B$ has a vertex incident upon at least $m$ whole-loops.
\end{theorem}

\begin{proof}
Lemma~6.7 of \cite{friedman_alon}
shows that if $u,v$ are distinct vertices in a graph, $\psi$, joined by an
edge, $e$, then the graph, $\psi'$ obtained by identifying $u$ and $v$ and
discarding $e$ has the same order as $\psi$ and
satisfies $\mu_1(\psi')\ge \mu_1(\psi)$.
By repeatedly performing this operation on a connected graph, $\psi$,
we get a graph $\psi'$ with one
vertex, of the same order as $\psi$ and with $\mu_1(\psi')\ge \mu_1(\psi)$.
It follows that if $\psi$ is a $(\ge\nu,<r)$-tangle, then so is
$\psi'$.  
If $\psi'$ is $d'$-regular, then since $\psi'$ has one vertex
we have $\mu_1(\psi')=d'-1$ (by the Ihara Determinantal formula\footnote{
  One can also see $\mu_1(\psi')=d'-1$ by noting that
  any non-backtracking walk can be augmented by one step
  that continues the walk to be non-backtracking
  in $d'-1$ ways, and can be made to be
  SNBC in $d'-2$ ways, 
  which shows that $d'(d'-1)^{k-2}(d'-2)\le \SNBC(\psi',k)\le d'(d'-1)^{k-1}$.
}).
If $\psi'$ has $m$ whole-loops and $m'$ half-loops, then
$\ord(\psi')=m+m'-1$ and $d'=2m+m'$.
It follows that if $m=m(B)$ is the smallest integer for which
$$
2 m - 1 > \mu_1^{1/2}(B),
$$
then $\tau_{\rm tang}\ge 2m-1$, and that equality holds if
the graph with $m$ whole-loops occurs in $\cC_n(B)$.
Then \eqref{eq_m_d_lower_bound_d} follows, and therefore
\eqref{eq_tau_tang_lower_bound_d} as well.
\end{proof}


\subsection{Lower Bounds on $\tau_{\rm tang}$ for Models Where
Whole-Loops Do Not Occur}

Similarly we can use the methods of Chapter~6 of \cite{friedman_alon}
to determine $\tau_{\rm tang}$ in cases where no graph with 
whole-loops occurs in $\cC_n(B)$.

\begin{theorem}\label{th_tau_tang_no_whole}
Let $\cC_n(B)$ be an algebraic model over a graph
$B$ such that no graph with one or more whole-loops occurs
in $\cC_n(B)$.
Then if $m'=m'(B)$ is the smallest integer with
$$
m' - 1 > \mu_1^{1/2}(B),
$$
then $\tau_{\rm tang}\ge m'-2$, and equality holds 
if some vertex of $B$ is incident upon $m'+1$ self-loops
(which may be any combination of whole-loops and half-loops).
In particular, for the
full cycle-involution model, of even or of odd degree, $\cC_n(B)$,
we have $m'=m'(B)$ depends only on $d$ and is given by
$$
m'(d) = \bigl\lfloor (d-1)^{1/2} \bigr\rfloor + 2,
$$
and
$$
\tau_{\rm tang} \ge \bigl\lfloor (d-1)^{1/2} \bigr\rfloor.
$$
\end{theorem}
\begin{proof}
Lemma~6.9 of \cite{friedman_alon} shows that if $\psi$
is graph with two vertices, $u,v$, of distance exactly
two, so that they are both adjacent to some vertex $w$,
then the graph, $\psi'$, obtained by identifying $u,v$ and discarding
one edge from $w$ to $u$ (or to $v$), has
$\mu_1(\psi')\ge \mu_1(\psi)$ (and, of course, $\ord(\psi')=\ord(\psi)$).
Note that this operation doesn't create any new self-loops.
Repeated application of this process yields a graph, $\psi'$,
with the same number of whole-loops and half-loops as in
$\psi$, such that $\psi',\psi$ have the same order and
$\mu_1(\psi')\ge \mu_1(\psi)$.
If $\psi'$ contains a pair of vertices $u,v$ that are joined by only one edge,
then the operation of Lemma~6.7 identifying $u$ and $v$ and
discarding the edge between them (see the proof of 
Theorem~\ref{th_tau_tang_whole} above)
produces a graph with one fewer vertices, the same order,
the same number of whole-loops and half-loops, and
no smaller a $\mu_1$.  Repeated application of this 
process yields a graph $\psi''$ with $\ord(\psi'')=\ord(\psi)$,
the same number of whole-loops and of half-loops as in $\psi$,
and $\mu_1(\psi'')\ge \mu_1(\psi)$.

Now say that $\psi$ is a $(\ge \nu,<r)$-tangle that
occurs in $\cC_n(B)$.  Then $\psi$ has
no whole-loops, and so the process above yields a $\psi''$ that 
is another such tangle and has no whole-loops.
Next we note:
\begin{enumerate}
\item if $\psi''$ has one vertex, then $\psi''$ is a bouquet of 
$d'$ half-loops, and 
$$
\mu_1(\psi'') = d'-1 = \ord(\psi'') ;
$$
\item if $\psi''$ has two vertices and $m$ edges, then
$$
\mu_1(\psi'') = m-1  = \ord(\psi'') + 1
$$
and $\mu_1(\psi'')$ is at most one less than maximum degree of a vertex,
which is at most $m-1$; furthermore this is attained for the
graph with two vertices joined by $m$ edges (i.e., without half-loops);
\item
if $\psi''$ has $n\ge 3$ vertices, then $\mu_1(\psi'')$ is at most
one less than the maximum degree of a vertex, which is at most
$$
\# E_{\psi''} - \binom{n-1}{2} 2 - 1
$$
since each pair of vertices of $\psi''$ is joined by at least two edges
and there are no whole-loops so each edge incident upon a vertex
can contribute at most $1$ to its degree;
hence
$$
\mu_1(\psi'')\le \# E_{\psi''} - (n-1)(n-2)- 1 
$$
and since $\# E_{\psi''} = \ord(\psi'')-\#V_{\psi''}=\ord(\psi'')-n$, we
have
$$
\mu_1(\psi'') \le \ord(\psi'') + n - (n-1)(n-2)- 1
=\ord(\psi'')+1-(n-2)^2 .
$$
It follows that $\mu(\psi'') < \ord(\psi'')+1$ if $n\ge 3$.
\end{enumerate}
It follows that if $m'=m'(B)$ is the smallest integer with
$$
m' - 1 > \mu_1(B),
$$
then
$$
\tau_{\rm tang} = m' - 2,
$$
with equality if the graph with two vertices and $m'$ edges occurs
in $\cC_n(B)$.
\end{proof}

\subsection{The One Vertex Case}

The theorems proven so far (all drawn from \cite{friedman_alon}) 
are sufficient
to determine $\tau_{\rm tang}$ in our basic models for graphs
for a bouquet of $d/2$ whole-loops (for $d\ge 4$ and even)
and a bouquet of $d$ half-loops (for $d\ge 3$),
except for $d$ small (since $m(d),m'(d)$ are of order $d^{1/2}$).
In fact, we easily check that these theorems determine $tau_{\rm tang}$
in all cases except
the full cycle model over a bouquet of $2$ or $3$ whole-loops
(since whole-loops cannot occur in the full cycle model, and
a graph with two vertices joined by, respectively, $3$ or $4$ edges
does not occur in the model).

In \cite{friedman_alon} these cases are dealt with by producing
graphs that occur in these models that match the lower bounds.
For example, the proof of Theorem~6.10 shows that for $3$ whole-loops,
the graph $\psi$ 
with three vertices, $\{v_1,v_2,v_3\}$, where $v_1,v_2$ are
joined by three edges and $v_2,v_3$ by two edges
has order $2$ and $\mu_1\ge\sqrt{6}>\sqrt{5}$; 
since $\psi$ occurs in this model,
we still conclude the bound in Theorem~\ref{th_tau_tang_whole}.
However, the example in the proof of Theorem~6.10 for the full cycle
model over a bouquet
of 2 whole-loops is not sufficient here
(since this example has $\mu_1(\psi)=\sqrt{3}$, while $\tau_{\rm tang}$---as
opposed to $\tau_{\rm fund}$ in \cite{friedman_alon}---looks for $\psi$
with $\mu(\psi_1)> \sqrt{3}$, the inequality required to be
strict\footnote{
  Certainly $\tau_{\rm tang}\le 2$ for the bouquet of $2$ whole-loops
  under the cyclic model: consider
  the graph $\psi$ with four vertices $\{v_1,v_2,v_3,v_4\}$
  where $v_i$ and $v_{i+1}$ are joined by two edges for $i=1,2,3$;
  we claim that $\mu_1(\psi)>\sqrt{3}$ (roughly since each time we 
  are in a middle vertex, i.e., $v_2,v_3$,
  we have three possible choices, and we will encounter side vertices
  less than half the time in a typical non-backtracking walk).
  Since $\ord(\psi)=2$ and $\psi$ occurs in this model.
  We believe that there is no $>\sqrt{3}$ tangle of order $1$ that
  occurs in this model, but there are a few cases to check.
}.

Note that Theorems~\ref{th_tau_tang_whole} and
\ref{th_tau_tang_no_whole} aren't strictly sufficient to
determine $\tau_{\rm tang}$ for an arbitrary bouquet of whole-loops and
half-loops.
However, the method of the proof of Theorem~\ref{th_tau_tang_no_whole}
shows gives a method of bounding the number of vertices, $n$, in a
$(>\mu_1^{1/2}(B))$-tangle, so for a fixed $B$ this would become a
finite procedure.

\subsection{The Case of $B$ with Sufficiently Large Girth}

In this section we note that for fixed $d$,
$\tau_{\rm tang}(B)$ for a $d$-regular graph, $B$, is bounded below
by a function of the girth of $B$ that tends to infinity as the
girth tends to infinity.
It follows that for fixed $d$, one can find $d$-regular graphs
where $\tau_{\rm tang}(B)$ is arbitrarily large.

\begin{theorem}
For a fixed $r\in\naturals$ and a real $\nu>1$, there is a 
$g\in\naturals$ such that any
$(\ge\nu,<r)$-tangle has girth at most $g$.
\end{theorem}
\begin{proof}
According to Article~III or Lemma~9.2 of \cite{friedman_alon},
there are a finite number of $(\ge\nu,<r)$-tangles
$\psi_1,\ldots,\psi_s$ such that any
$(\ge\nu,<r)$-tangle contains a subgraph isomorphic to
$\psi_i$ for some $i$.
Since $\nu\ge 1$, each $\psi_i$ contains a cycle of some length $L_i$.
Hence any $(\ge\nu,<r)$-tangle has girth at most $\max_i L_i$.
\end{proof}
Of course, the above proof does not give an explicit estimate
of $g$.

\begin{corollary}
For a fixed $r,d\in\naturals$, there is a $g$ such that if
$B$ is $d$-regular and of girth greater than $g$, then
$\tau_{\rm tang}\ge r$.
\end{corollary}
\begin{proof}
If $\psi$ occurs in any model over $B$, then $\psi$ admits an
\'etale map to $B$ and hence if $\psi$ has an SNBC closed walk of
length $k$, then so does $B$; it follows that the girth of $\psi$
is at least that of $B$.
Hence if $g$ is as in the above theorem
with $\nu=(d-1)^{1/2}$ and $r$ fixed, then the order of any
graph, $\psi$, occurring in any model of $B$ where $B$ is $d$-regular
(and therefore $\mu_1(B)=d-1$) and of girth greater than $g$
has $\ord(\psi)\ge r$.
\end{proof}

\providecommand{\bysame}{\leavevmode\hbox to3em{\hrulefill}\thinspace}
\providecommand{\MR}{\relax\ifhmode\unskip\space\fi MR }
\providecommand{\MRhref}[2]{%
  \href{http://www.ams.org/mathscinet-getitem?mr=#1}{#2}
}
\providecommand{\href}[2]{#2}

\end{document}